\documentclass[journal,twoside,web]{ieeecolor}
\usepackage{generic}
\usepackage{cite}
\usepackage{amsmath,amssymb,amsfonts}
\usepackage{graphicx}
\usepackage{hyperref}
\hypersetup{hidelinks=true}
\usepackage[acronym]{glossaries}
\usepackage{xcolor}
\usepackage{caption}
\usepackage{subcaption}

\def\wi{w_\mathrm{ini}}
\def\yi{y_\mathrm{ini}}
\def\ui{u_\mathrm{ini}}
\def\uf{u_\mathrm{f}}
\def\yf{y_\mathrm{f}}
\def\ur{u_\mathrm{ref}}
\def\yr{y_\mathrm{ref}}
\def\Ti{T_\mathrm{ini}}
\def\Tf{T_\mathrm{f}}
\def\Wd{W}
\def\wfree{w_\mathrm{free}}
\def\wdep{w_\mathrm{dep}}
\def\Se{\Sigma^\mathrm{e}}
\def\Seh{\hat{\Sigma}^\mathrm{e}}
\def\Sfull{\Sigma^\mathrm{w}}
\def\Sfullh{\hat{\Sigma}^\mathrm{w}}
\def\Sp{\Sigma_\mathrm{pred}}
\def\Sph{\hat{\Sigma}_\mathrm{pred}}

\newtheorem{thm}{Theorem}
\newtheorem{ass}{Assumption}

\newtheorem{prop}{Proposition}
\newtheorem{definition}{Definition}
\newtheorem{lemma}{Lemma}
\newtheorem{rmk}{Remark}

\newacronym{LTI}{LTI}{linear time-invariant}
\newacronym{SPC}{SPC}{\textit{subspace predictive control}}
\newacronym{DeePC}{DeePC}{data-enabled predictive control}
\newacronym{MPC}{MPC}{model predictive control}
\newacronym{LMI}{LMI}{linear matrix inequality}
\newacronym{HVAC}{HVAC}{heating, ventilation, and air conditioning}
\newacronym{EIV}{EIV}{errors-in-variables}
\newacronym{VaR}{VaR}{value at risk}
\newacronym{CVaR}{CVaR}{conditional value at risk}

\def\BibTeX{{\rm B\kern-.05em{\sc i\kern-.025em b}\kern-.08em
    T\kern-.1667em\lower.7ex\hbox{E}\kern-.125emX}}
\begin{document}
\title{Gaussian behaviors \\ and stochastic data-driven control}

\author{Andr\'as Sasfi, \IEEEmembership{Student Member, IEEE} Alberto Padoan, \IEEEmembership{Member, IEEE}, Ivan Markovsky, \IEEEmembership{Member, IEEE}, Florian D\"orfler, \IEEEmembership{Fellow, IEEE}
\thanks{A. Sasfi, and F. D\"orfler are with the Department of Information Technology and
Electrical Engineering, ETH Z\"urich, 8092 Z\"urich, Switzerland (e-mail: \{asasfi, doerfler\}@control.ee.ethz.ch).
A. Padoan is with the Department of Electrical \& Computer Engineering, University of British Columbia (e-mail: apadoan@ece.ubc.ca).
I. Markovsky is with the Catalan Institution for Research and Advanced Studies, 08010 Barcelona, Spain, and also with the Centre Internacional de M\`etodes Num\`erics en Enginyeria, 08034 Barcelona, Spain (e-mail: imarkovsky@cimne.upc.edu). This work was supported by the SNF/FW WEAVE Project 200021E\_20397 and the FWO WEAVE Project G033822N}
}

\maketitle

\begin{abstract}
We propose a stochastic behavioral modeling framework, termed Gaussian behaviors, which augments a deterministic \gls{LTI} behavior with a Gaussian noise component.
We show that this notion is a tractable subclass of stochastic behaviors and encompasses classical parametric stochastic \gls{LTI} state-space system models as special cases. 
Analogously to deterministic \gls{LTI} behaviors, the framework enables simple and tractable stochastic data-driven control methods. 
To this end, we obtain a method for prediction by conditioning the Gaussian behavior on the known part of the trajectory, which is identified directly from the sample covariance of trajectory data.
Building on this method, we develop predictive control formulations that optimize over feedforward or disturbance affine feedback policies.
The resulting formulations are shown to be convex.
We further derive a finite-sample confidence bound on the prediction accounting for both aleatoric and epistemic uncertainty, and incorporate it into a robust control method, for which a tractable convex upper bound is obtained.
Within this framework, subspace predictive control is recovered when only the mean prediction is used, while data-enabled predictive control is shown to account for the prediction uncertainty in an optimistic fashion.
Numerical case studies illustrate the benefits of the proposed methods.
\end{abstract}


\section{Introduction} \label{sec:introduction}
Data-driven predictive control~\cite{coulson2019data,coulson2021distributionally,berberich2020data,breschi2023data,depersis2019formulas} has recently emerged as an alternative to \gls{MPC}~\cite{rawlings2017model}.
By bypassing the need for an explicitly identified model, which is often the most expensive step of the control design process~\cite{dorfler2022bridging}, these methods synthesize control policies directly from measured data.
The theoretical backbone for these approaches is behavioral systems theory~\cite{willems1997introduction}.
Instead of relying on a particular representation, behavioral systems theory describes systems as sets of trajectories, which correspond to shift-invariant subspaces in the \acrfull{LTI} setting.
This perspective enables a direct data-driven representation of \gls{LTI} systems using matrices of time-series data~\cite{willems2005note}, making it particularly suited for formulating data-driven control methods.
When augmented with regularization, these algorithms perform well even in scenarios involving large amounts of uncertainty~\cite{dorfler2022bridging,markovsky2021behavioral}.
However, because the underlying behavioral model is fundamentally deterministic, most existing works a posteriori robustify the control design~\cite{huang2021quadratic,coulson2021distributionally,berberich2020data}.
This discrepancy highlights the need for extending behavioral systems theory to the stochastic setting, which is also emphasized in the  recent survey~\cite{markovsky2021behavioral}. 
Thus, the goal of this work is to provide such an extension while enabling the design of simple and pragmatic stochastic data-driven control formulations, similarly to the deterministic case.

An early attempt to establish a stochastic extension of behavioral systems theory is~\cite{willems2012open}, which is limited to static systems.
This work was subsequently extended to dynamic systems in~\cite{baggio2017lti}.
While both works are mathematically elegant, they are not immediately useful for and have also not successfully translated into actionable stochastic data-driven control formulations.
Alternatively,~\cite{faulwasser2023behavioral,pan2022stochastic} extend behavioral systems theory to the stochastic domain using polynomial chaos expansions.
However, the computational complexity of these techniques increases rapidly with the order of the expansion, which limits their scalability for practical applications.

Beyond the behavioral perspective, stochastic data-driven control has been broadly addressed using parametric models such as input-output~\cite{chiuso2025harnessing} or state-space representations~\cite{fiedler2023probabilistic}.
In~\cite{fiedler2023probabilistic}, a probabilistic multi-step prediction model is identified via maximum likelihood techniques.
This approach quantifies parametric uncertainty arising from finite data and employs the resulting predictive distribution within a chance-constrained stochastic \gls{MPC} framework. 
Alternatively,~\cite{chiuso2025harnessing} derives the expected control cost conditioned on the data in closed form.
This work establishes a separation principle where the ARX parameter vector and its associated uncertainty are estimated first, followed by the minimization of a cost function comprising a certainty equivalence term and an uncertainty induced regularizer.
Despite their contributions, these methods have the following limitations. 
First, both approaches are fundamentally risk neutral, minimizing the expected control cost without accounting for the variance or worst case realizations of the objective.
Furthermore, the work in~\cite{chiuso2025harnessing} enforces constraints solely on the mean prediction, agnostic of any risk.
Second, both frameworks optimize over feedforward input sequences exclusively, which leads to inherently conservative predictions.
Similar limitations apply to the signal matrix model approach of~\cite{smith2024optimal}, which derives the best linear unbiased multi-step predictor and its covariance directly from data, yet employs only the mean prediction with feedforward inputs.
In traditional stochastic MPC, this conservatism is often mitigated by optimizing over disturbance affine feedback policies~\cite{goulart2006optimization}.
Yet, this technique remains largely unexplored in data-driven contexts.
The notable exception is~\cite{li2024distributionally}, which requires known noise covariance contradicting the core motivation behind data-driven control design.
The proposed stochastic control methods, relying on our novel stochastic behavioral framework, address these limitations.

The main contributions of our work are the following.
\begin{itemize}
    \item First, we introduce a pragmatic and parsimonious stochastic modeling framework within the behavioral setting, termed \textit{Gaussian behaviors}.
    A Gaussian behavior consists of a deterministic \gls{LTI} behavior that describes all nominal trajectories, alongside a Gaussian noise component to capture random deviations.
    We show that this framework constitutes a tractable subclass of abstract stochastic behaviors defined in~\cite{willems2012open,baggio2017lti}, while encompassing existing parametric stochastic modeling frameworks.

    \item 
    Second, we derive a simple linear method for prediction, identified directly from the sample covariance of trajectory data by conditioning on the known portion of the trajectory.
    While a Gaussian behavior is not uniquely identifiable from data, this predictive distribution is.
    Its mean corresponds to the established subspace predictor~\cite{favoreel1999spc}, while its covariance captures the prediction uncertainty. 
    In addition to this \emph{aleatoric} uncertainty, we also characterize the \emph{epistemic} uncertainty arising from finite data, following a logic similar to~\cite{fiedler2023probabilistic}.
    Finally, we demonstrate that this framework can seamlessly incorporate uncertain forecasts of external variables, which is a crucial benefit for predictive control applications.
   
    \item Third, we propose a stochastic data-driven control formulation that leverages the aleatoric uncertainty in the Gaussian behavior and optimizes over disturbance affine feedback policies.
    We show that this stochastic scheme can be reformulated as a deterministic convex optimization problem.
    This approach leads to a fully data-driven method that reduces the conservatism of predictions.
    The practical benefit of using disturbance affine feedback policies is showcased on a building control case study.
    
    \item
    Fourth, we propose a robust formulation that explicitly accounts for both aleatoric and epistemic uncertainty in the prediction by minimizing the worst-case cost over a confidence ellipsoid.
    The resulting min-max problem provides an upper bound on the \gls{VaR} of the control cost.
    Through the use of Lagrangian duality, we obtain a tractable convex surrogate that overapproximates the original min-max problem, leading to an efficient and robust data-driven control formulation.
    Conversely, we show that an optimistic standpoint, leading to a min-min problem, recovers the \gls{DeePC} method~\cite{coulson2019data} with projection-based regularization, endowing regularization in data-driven control with a stochastic interpretation.
    A case study confirms that our robust method yields a more cautious controller than existing schemes.
\end{itemize}

This article expands upon our preliminary work in~\cite{sasfi2025gaussian} with several changes and additions.
The definition of Gaussian behaviors has been modified to align better with established stochastic behavioral notions.
Compared to~\cite{sasfi2025gaussian}, the system identification section has been expanded by the finite-sample confidence bound on the prediction accounting for epistemic uncertainty, and the methodology for incorporating uncertain forecasts. 
The stochastic control formulations have also been extended to include chance constraints and optimization over disturbance affine feedback policies. 
Moreover, the robust formulation has been reworked to rely on the finite-sample confidence bound, making it fundamentally different from the preliminary version. 
Finally, opposed to~\cite{sasfi2025gaussian}, we provide multiple case studies to illustrate the practical benefits of our framework.

The remainder of the paper is organized as follows.
Section~\ref{sec:prelim} introduces notation and reviews behavioral systems theory, stochastic state-space models, and data-driven predictive control. 
Section~\ref{sec:Gaussian_behaviors} defines Gaussian behaviors and relates them to existing modeling frameworks.
Section~\ref{sec:GB_sysID} addresses identification and prediction, including the finite-sample uncertainty bound and the treatment of uncertain forecasts. 
Section~\ref{sec:CE_control} develops chance-constrained formulations under aleatoric uncertainty, both with feedforward inputs and with disturbance affine feedback policies. 
Section~\ref{sec:uncertain_control} presents the robust formulation under both aleatoric and epistemic uncertainty.
Section~\ref{sec:case_studies} reports the numerical case studies, and Section~\ref{sec:conclusion} concludes the paper.

\section{Preliminaries} \label{sec:prelim}
We first provide preliminaries on notation, stochastic state-space models, behavioral systems theory, and data-driven predictive control for linear systems.

\subsection{Notation}
The probability of an event and the expectation of a random variable are denoted by $\mathrm{Pr}(\cdot)$ and $\mathbb{E}[\cdot]$, respectively.
The multivariate Gaussian distribution with mean $\mu$ and covariance matrix $\Sigma$ is denoted by $\mathcal{N}(\mu, \Sigma)$.
The fact that a symmetric matrix $P$ is positive definite (semi-definite), is denoted by $P \succ 0$ ($P \succeq 0$).
For $P \succeq 0$, the weighted semi-norm of a vector $x$ is $\|x\|_P^2 = x^\top P x$, which is a norm if $P \succ 0$.
The trace of a matrix is denoted by $\mathrm{tr}(\cdot)$, the image (column space) by $\mathrm{im}(\cdot)$ and the Moore–Penrose pseudoinverse by $(\cdot)^\dagger$.

\subsection{Stochastic LTI state-space systems}
Consider a discrete-time stochastic \gls{LTI} state-space system governed by the equations
\begin{align} \label{eq:stochastic-SS}
\begin{split}
    x_{t+1} &= Ax_t + Bu_t + \xi_t,\\
    y_t & = Cx_t + D u_t + \eta_t,
\end{split}
\end{align}
where $x_t,~u_t$, and $y_t$ are the state, input, and output; and the random vectors $\xi_t\in\mathbb{R}^n$ and $\eta_t\in\mathbb{R}^p$ represent the process and measurement noise, respectively.
We assume that the sequences $\{\xi_t\}$ and $\{\eta_t\}$ are independent and identically distributed with zero mean and respective covariance matrices $\Sigma^\xi$ and $\Sigma^\eta$.
Here, the initial state $x_0$ and the input $u$ may also be random.
Over a finite time horizon of length $L$, the stacked output sequence $y=\begin{bmatrix}y_t^\top&y_{t+1}^\top&\dots&y_{t+L-1}^\top\end{bmatrix}^\top$ forms a random vector.
It can be expressed as an affine function of the state $x_t$ alongside the stacked trajectories of the inputs $u$ and the noise realizations $\xi$ and $\eta$ (defined analogously to $y$):
\begin{align} \label{eq:stochastic_ss_linear}
    y = \mathcal{O}_L x_t + \mathcal{T}_L^\mathrm{u} u + \mathcal{T}_L^\xi \xi + \eta.
\end{align}
Here, the extended observability matrix $\mathcal{O}_L$ and the lower-triangular Toeplitz matrix $\mathcal{T}_L^\mathrm{u}$ are defined as
\small
\begin{align*}
    \mathcal{O}_L = \begin{bmatrix}
        C \\ CA \\ \vdots \\ CA^{L-1}
    \end{bmatrix},~ \mathcal{T}_L^\mathrm{u} = 
    \begin{bmatrix}
      D & 0 & 0 & \dots & 0 \\
        CB & D & 0 & \dots & 0 \\
        CAB & CB & D &\dots  & 0\\
        \vdots & \vdots & \ddots & \ddots & \vdots \\
        CA^{L-2}B & \dots & \dots & CB & D   
  \end{bmatrix},
\end{align*}
\normalsize
while the matrix $\mathcal{T}_L^\xi$ is constructed identically to $\mathcal{T}_L^\mathrm{u}$, but with $B$ replaced by the identity matrix and $D$ replaced by $0$.

\subsection{Behavioral systems theory for deterministic LTI systems} \label{sec:BST}
In the behavioral framework, a dynamical system is defined by the set of its trajectories, called the \textit{behavior}, without relying on a specific representation.
We consider discrete-time system trajectories of finite length $L > 0$, defined as a collection of signals $w_t$ (typically consisting of inputs and outputs) as
$$
w_{[t,t+L-1]} = (w_t~w_{t+1}~\dots~w_{t+L-1}),
$$
or simply $w$ for brevity.
For an \gls{LTI} system, the restricted behavior, $\mathcal{B}_L$, comprises all valid length-$L$ trajectories, informally defined as
\begin{align*}
    \mathcal{B}_L = \{w\in \mathbb{R}^{qL}~|~w \text{ is a length-$L$ trajectory of the system}\},
\end{align*}
For a complete \gls{LTI} system, the behavior is a shift-invariant subspace within the ambient space of all potential trajectories~\cite{markovsky2022data}.
Provided that the horizon $L$ is sufficiently large, $\mathcal{B}_L$ is a subspace of dimension $d = mL + n$ in $\mathbb{R}^{qL}$, where $m$ and $n$ denotes the number of inputs and the order of the system~\cite{willems1997introduction}, respectively.

A basis for $\mathcal{B}_L$ can be derived from a deterministic \gls{LTI} state-space model~\cite{markovsky2022data}.
Alternatively, deterministic \gls{LTI} behaviors can be represented purely from sufficiently rich data.
Suppose we collect $D > qL$ length-$L$ trajectories from the system, denoted as $w^\mathrm{d}_i\in\mathbb{R}^{qL}$ for $i\in\{1,\dots,D\}$, and arrange them into the data matrix $\Wd = [w^\mathrm{d}_1~\dots~w^\mathrm{d}_D]$.
As established by an extension of the Fundamental Lemma~\cite{van2020willems}, if the input sequences within this dataset are \textit{collectively persistently exciting} and the system is controllable, then the matrix $\Wd$ has rank $mL + n$.
Consequently, the column space of the data matrix spans the restricted behavior, that is, $\mathcal{B}_L = \mathrm{im}(\Wd)$.

\subsection{Data-driven predictive control formulations} \label{sec:prelim-dd-control}
Several data-driven predictive control approaches leverage the pragmatic data representation of the behavior, namely $\mathcal{B}_L = \mathrm{im}(\Wd)$~\cite{coulson2019data,favoreel1999spc,breschi2023data,depersis2019formulas}. 
Without loss of generality, any trajectory $w_t \in \mathbb{R}^q$ of a complete \gls{LTI} system can be partitioned into an input component $u_t \in \mathbb{R}^m$ and an output component $y_t \in \mathbb{R}^{p}$, with $q=m+p$.
To incorporate real-time measurements, the input and output trajectories of length $L = \Ti + \Tf$ are divided into an initial and a future segment of length $\Ti$ and $\Tf$, denoted by $\ui,~\yi$ and $\uf,~\yf$, respectively. 
The initial segment $\wi = \begin{bmatrix}
    \ui^\top & \yi^\top
\end{bmatrix}^\top$ is known and the future part is optimized.
The full trajectory is then $w = \begin{bmatrix} \wi^\top & \uf^\top & \yf^\top\end{bmatrix}^\top$. 
The predictive control objective is typically formulated as the minimization of the quadratic cost
\begin{align*}
    J(\uf,\yf) = \|\uf-\ur\|_R^2 + \|\yf-\yr\|_Q^2,
\end{align*}
where $\ur$ and $\yr$ denote reference sequences, while $R\succ 0$ and $Q\succ 0$ are weight matrices for the future input and output sequences, respectively.
In addition, the future input and output trajectories are often required to satisfy constraints defined by a non-empty, closed, and convex sets $\mathcal{U}$ and $\mathcal{Y}$.

We briefly summarize two established data-driven predictive control frameworks.
To this end, let us partition the historical data matrix $\Wd$ into block rows corresponding to the components $\wi$, $\uf$, and $\yf$, yielding $\Wd = \begin{bmatrix} W_\mathrm{p}^\top & U_\mathrm{f}^\top & Y_\mathrm{f}^\top\end{bmatrix}^\top$. 
The \gls{SPC} scheme~\cite{favoreel1999spc,markovsky2022data} is formulated as
\begin{align} \label{eq:SPC}
\begin{split}
    \min_{\uf\in \mathcal{U},\yf \in \mathcal{Y}} & \quad  J(\uf,\yf) \\
    \mathrm{s.t.} & \quad \yf = Y_\mathrm{f} \begin{bmatrix}
        W_\mathrm{p} \\ U_\mathrm{f}
    \end{bmatrix}^{\dagger} \begin{bmatrix}
        \wi \\ \uf
    \end{bmatrix},
\end{split}
\end{align}
where the constraint is known as the \textit{subspace predictor}.
Alternatively, \gls{DeePC}~\cite{dorfler2022bridging} solves the optimization problem
\begin{align} \label{eq:DeePC}
\begin{split}
    \min_{\uf \in \mathcal{U},\yf\in\mathcal{Y},g\in\mathbb{R}^{D}} & \quad J(\uf,\yf) + \lambda_g \cdot h(g) \\
    \mathrm{s.t.} & \quad \begin{bmatrix}
        \wi \\ \uf \\ \yf
    \end{bmatrix} = \begin{bmatrix}
        W_\mathrm{p} \\ U_\mathrm{f}\\ Y_\mathrm{f}
    \end{bmatrix} g,
\end{split}
\end{align}
where $h(g)$ is a regularizer scaled by the regularization parameter $\lambda_g \geq 0$.
Standard choices for $h(g)$ include the $\ell_1$-norm, the squared $\ell_2$-norm, and the projection-based regularizer~\cite{markovsky2022data,dorfler2022bridging} 
\begin{align} \label{eq:proj-based-reg}
h(g) = \left\|\left(I-\begin{bmatrix}
    W_\mathrm{p} \\ U_\mathrm{f}
\end{bmatrix}^\dagger \begin{bmatrix}
    W_\mathrm{p} \\ U_\mathrm{f}
\end{bmatrix}\right)g\right\|_2^2.
\end{align}

\section{Gaussian Behaviors} \label{sec:Gaussian_behaviors}
We present a pragmatic and parsimonious modeling framework termed \textit{Gaussian behavior} for discrete-time stochastic dynamical systems that is amenable to data-driven control, followed by comparisons with other stochastic models.

\subsection{Definition}
In the spirit of behavioral systems theory, our approach relies on a trajectory-based perspective.
We incorporate stochasticity as a fundamental layer over a nominal subspace representation of an \gls{LTI} system originating from deterministic behavioral systems theory.
To ensure practical applicability and enable tractable control formulations, we consider this stochastic innovation to be stationary and normally distributed.
In fact, our proposed Gaussian behaviors are a tractable subclass of stochastic behaviors introduced in~\cite{willems2012open,baggio2017lti}, see Subsection~\ref{sec:GB_relations_BST} below.
Because the problem setting is finite horizon, we bypass infinite-horizon technicalities and provide the definition directly for finite-length trajectories, motivating the following formulation of Gaussian behaviors.

\begin{definition} \label{def:GB}
A \textit{(finite-horizon \gls{LTI}) Gaussian behavior} is a pair $(\mathcal{B}_L,\Se)$, where
\begin{itemize}
    \item $\mathcal{B}_L\subseteq \mathbb{R}^{qL}$ is a deterministic finite-horizon \gls{LTI} behavior as defined in Section~\ref{sec:BST}, and
    \item $\Se \in \mathbb{R}^{qL\times qL}$ is a positive semi-definite covariance matrix that describes the exogenous stochasticity in the trajectories.
\end{itemize}
Furthermore, any random length-$L$ trajectory $w_{[t,t+L-1]}$ of the stochastic system has a decomposition
$$
w_{[t,t+L-1]} = \overline{w}_{[t,t+L-1]} + e,
$$
where $\overline{w}_{[t,t+L-1]}\in\mathcal{B}_L$ and $e\sim \mathcal{N}(0,\Se)$.
\end{definition}

According to Definition~\ref{def:GB}, any trajectory is naturally decomposed into a nominal component $\overline{w}$ and a stochastic deviation $e$.
As in the classical deterministic behavioral systems theory~\cite{willems1997introduction}, the nominal trajectory $\overline{w}$ is constrained to the deterministic behavior (the linear subspace $\mathcal{B}_L$).
The nominal component $\overline{w}\in\mathcal{B}_L$ may itself be stochastic, but its probability law is not defined by the Gaussian behavior.
This randomness is not specific to the system, as it originates from random inputs (e.g., for excitation) or initial conditions (e.g., in an estimation or covariance steering setting).
This type of stochasticity does not cause deviations from the nominal behavior, i.e.,  $\overline{w}$ is supported on $\mathcal{B}_L$ for any starting time $t$.
We will refer to this randomness as the \textit{endogenous} stochasticity in the trajectories.

The other component of the Gaussian behavior, $\Se$, describes what we call the \textit{exogenous} stochasticity, which quantifies the randomness that causes the trajectories to deviate from the nominal subspace.
The decomposition $w_{[t,t+L-1]} = \overline{w}_{[t,t+L-1]} + e$ implies that the same noise perturbation is added to the nominal component that can be anywhere on the nominal subspace $\mathcal{B}_L$.
Thus, this decomposition suggests that the stochasticity is defined on the quotient space $\mathbb{R}^{qL}/\mathcal{B}_L$, similarly to~\cite{willems2012open,baggio2017lti}.
Furthermore, the statistics of $e$ are independent of the starting time $t$ of the trajectory, i.e., $e\sim\mathcal{N}(0,\Se)$ for any $t\geq0$.
For example, when using a state-space representation of the stochastic system, the exogenous stochasticity corresponds to process and measurement noise.
By capturing the system's exogenous noise in this manner, the model explicitly quantifies uncertainty while maintaining a simple linear-algebraic structure, thereby enabling the tractable implementation of downstream tasks such as predictive control, see Sections~\ref{sec:CE_control} and \ref{sec:uncertain_control}.

The decomposition into a nominal part $\overline{w} \in \mathcal{B}_L$ and a noise term $e \sim \mathcal{N}(0, \Se)$ is not unique, since the subspace  $\mathcal{B}_L$ and the support of $e$ may intersect.
Correspondingly, a Gaussian behavior is not uniquely identifiable from data, which we resolve in Section~\ref{sec:GB_sysID}.
The proposed formulation is consistent with established modeling frameworks for stochastic dynamical systems, as we demonstrate next.

\subsection{Relation to stochastic LTI state-space representations} \label{sec:GB_relations_SS}
We now show that a Gaussian behavior can be constructed from a stochastic state-space representation. 
Assuming the input-output partitioning\footnote{Here, we order the trajectory as $w = \begin{bmatrix}
    u^\top & y^\top
\end{bmatrix}^\top$, which is inconsistent with the ordering used in previous sections.
Formally, we work with the permuted trajectory $\Pi w$, where $\Pi$ is a permutation matrix, but for simplicity we continue to denote it by $w$.}
$w = \begin{bmatrix}
    u^\top & y^\top
\end{bmatrix}^\top$, \eqref{eq:stochastic_ss_linear} shows that the entire trajectory can be expressed as
\begin{align} \label{eq:w_state_space}
w = \underbrace{\begin{bmatrix} I & 0 \\ \mathcal{T}_L^\mathrm{u} & \mathcal{O}_L \end{bmatrix} \begin{bmatrix} u \\ x_t \end{bmatrix}}_{\overline{w}} + \underbrace{\begin{bmatrix} 0 \\ \mathcal{T}_L^\xi \xi + \eta \end{bmatrix}}_{e}.
\end{align}
If the representation~\eqref{eq:stochastic-SS} is observable, $\overline{w}$ is restricted to a subspace of dimension $d=mL+n$~\cite{markovsky2022data}. 
Assuming $L$ is larger than the lag of the system~\cite{markovsky2022data}, this subspace corresponds to the deterministic \gls{LTI} behavior, i.e.,
$$
\mathcal{B}_L = \mathrm{im}\left( \begin{bmatrix} I & 0 \\ \mathcal{T}_L^\mathrm{u} & \mathcal{O}_L \end{bmatrix}\right).
$$
In the decomposition~\eqref{eq:w_state_space}, the nominal component $\overline{w}$ is stochastic due to randomness in the initial state (and possibly in the input), yet it is supported on the nominal subspace $\mathcal{B}_L$.
Furthermore, the exogenous noise component $e$ is a zero-mean Gaussian random vector with covariance matrix
$$
\Se = \begin{bmatrix}
            0 & 0 \\
            0 & \mathcal{T}_L^\xi(I_L \otimes \Sigma^\xi)(\mathcal{T}_L^\xi)^\top + (I_L\otimes \Sigma^\eta)
        \end{bmatrix},
$$
where $\otimes$ and $I_L$ denotes the Kronecker product and the identity matrix of size $L\times L$, respectively.
This covariance is independent of the trajectory's starting time $t$, as required by Definition~\ref{def:GB}. Note that the single-step model~\eqref{eq:stochastic-SS} imposes a special structure on $\Se$. 
Furthermore, if $\Se$ has a block-Toeplitz structure, the exogenous component $e \sim \mathcal{N}(0, \Se)$ is a length-$L$ window of a zero-mean stationary Gaussian process via Kolmogorov's extension theorem~\cite{papoulis1991probability}.
Importantly, we do not impose any structure on $\Se$ other than positive semi-definiteness in Definition~\ref{def:GB}, which yields a more expressive framework, as illustrated in Sections~\ref{sec:forecast} and~\ref{sec:HVAC}.

\subsection{Relation to deterministic and stochastic LTI behaviors} \label{sec:GB_relations_BST}
In the deterministic setting, the nominal part of the trajectory is deterministic in Definition~\ref{def:GB}, and it is restricted to the subspace, that is, $\overline{w}\in\mathcal{B}_L$.
Therefore, deterministic \gls{LTI} behaviors can be trivially recovered by setting the exogenous stochasticity to zero, i.e., $\Se = 0$.

Definitions of stochastic systems in the behavioral setting were established in~\cite{willems2012open} and~\cite{baggio2017lti}.
The seminal work~\cite{willems2012open} treats \textit{static} systems only, while~\cite{baggio2017lti} provides an extension of the definition to the \textit{dynamical} setting. 
Since our work focuses on control applications later on, we relate our work to stochastic \textit{dynamical} systems.
More specifically, we show that our notion of Gaussian behaviors defines a special class of stochastic dynamical systems as defined in~\cite{baggio2017lti}.

The stochastic behavioral framework established in~\cite{baggio2017lti} (building on the notions in~\cite{willems2012open}) integrates a probability space structure into the set of possible system trajectories.
This general approach encompasses both deterministic behaviors and conventional stochastic processes~\cite{doob1953stochastic} as specific instances.
In this context, a \emph{stochastic dynamical system} is formally defined by the quadruple 
$$
(\mathbb{T}, \mathbb{W}, \mathcal{F}, \mathcal{P}),
$$
where:
\begin{itemize}
    \item $\mathbb{T}$ is the \emph{time axis} (e.g., $\mathbb{T} = \mathbb{Z}$ or $\mathbb{T} = \mathbb{R}$),
    \item $\mathbb{W}$ is the \emph{signal set} (e.g., $\mathbb{W} = \mathbb{R}^q$ or $\mathbb{W} = \mathbb{C}^q$),
    \item $\mathcal{F}$ represents a \emph{$\sigma$-algebra} of events, which is a collection of measurable subsets of the trajectory space $\mathbb{W}^\mathbb{T}$,
    \item $\mathcal{P}$ is a \emph{probability measure} on $(\mathbb{W}^\mathbb{T}, \mathcal{F})$ that assigns probabilities to events, which are entire sets of trajectories.
\end{itemize}
Building upon this general framework, an LTI stochastic process is defined by restricting the system to a discrete time axis ($\mathbb{T} = \mathbb{Z}$) and a real-valued finite-dimensional signal space ($\mathbb{W} = \mathbb{R}^q$).
The defining characteristic of this process is the existence of a deterministic, linear time-invariant behavior $\mathcal{B} \subset (\mathbb{R}^q)^{\mathbb{Z}}$, which serves as the \emph{fiber} of the system.
Under this formulation, the $\sigma$-algebra of events is constructed from the Borel subsets of the quotient space $(\mathbb{R}^q)^{\mathbb{Z}}/\mathcal{B}$, and the probability measure $\mathcal{P}$ is defined directly over this quotient space.
The meaning of defining the probability measure on the quotient space was made more concrete in both~\cite{willems2012open} and~\cite{baggio2017lti} by considering a kernel representation $\mathbf{R}(\sigma)$ of the behavior $\mathcal{B}$~\cite{markovsky2022data}.
Then, random trajectories of the system satisfy
\begin{align} \label{eq:baggio_kernel_repr}
\mathbf{R}(\sigma)\cdot w_t = \varepsilon_t,
\end{align}
where $\mathbf{R}(\sigma)$ is a kernel representation of the behavior~\cite{markovsky2022data}, and $\varepsilon_t$ is a stochastic process.

Our proposed definition of a Gaussian behavior serves as a practical, finite-length specialization of this concept.
Specifically, we avoid the abstract measure-theoretic technicalities associated with defining $\sigma$-algebras by restricting the stochasticity to Gaussian distributions and implicitly assuming a Borel $\sigma$-algebra.
In our framework, the restricted deterministic \gls{LTI} behavior $\mathcal{B}_L$ assumes the role of the fiber.
Definition~\ref{def:GB} can be read as an admissibility condition on the law of a trajectory: letting $\mu$ denote the law of a length-$L$ trajectory $w$, we say $\mu$ is admissible for $(\mathcal{B}_L, \Se)$, if there exist $\overline{w} \in \mathcal{B}_L$ a.s. and $e \sim \mathcal{N}(0, \Se)$ with $\overline{w} + e \sim \mu$.
The following proposition shows that this decomposition is equivalent to the kernel formulation~\eqref{eq:baggio_kernel_repr} used in~\cite{willems2012open, baggio2017lti}.

\begin{prop} \label{prop:kernel_equivalence}
Let $(\mathcal{B}_L, \Se)$ be a Gaussian behavior and let
$\mathbf{R} \in \mathbb{R}^{(qL-d) \times qL}$ be a kernel representation of the restricted behavior $\mathcal{B}_L$.
A trajectory $w \in \mathbb{R}^{qL}$ admits the decomposition of Definition~\ref{def:GB} if and only if
\begin{align*}
\mathbf{R} w = \varepsilon, \qquad \varepsilon \sim \mathcal{N}\!\left(0,\, \mathbf{R} \Se \mathbf{R}^\top\right).
\end{align*}
\end{prop}
\vspace{2mm}
\begin{proof}
($\Rightarrow$) If $w = \overline{w} + e$ with $\overline{w} \in \mathcal{B}_L$ and $e \sim \mathcal{N}(0, \Se)$, then $\mathbf{R}\overline{w} = 0$, so
$\mathbf{R}w = \mathbf{R}e \sim \mathcal{N}(0, \mathbf{R}\Se \mathbf{R}^\top)$.\\
($\Leftarrow$) Given any $w$ with $\mathbf{R}w \sim \mathcal{N}(0, \mathbf{R}\Se \mathbf{R}^\top)$, set $S := \mathbf{R}\Se \mathbf{R}^\top$, take $\zeta \sim \mathcal{N}(0,\ \Se - \Se \mathbf{R}^\top S^\dagger \mathbf{R} \Se)$ independent of $w$, and set
$$
e := \Se \mathbf{R}^\top S^\dagger (\mathbf{R}w) + \zeta, \qquad \bar{w} := w - e.
$$
Then $e \sim \mathcal{N}(0,\Se)$, $\mathbf{R}e = \mathbf{R}w$ a.s. (hence $\bar{w} \in \mathcal{B}_L$ a.s.), and $\bar{w} + e = w$.
\end{proof}

Finally, note that, using the terminology in~\cite{willems2012open}, Definition~\ref{def:GB} describes an \textit{open} system, as no restrictions on the stochasticity of $\overline{w}$ are imposed other than being supported on $\mathcal{B}_L$.
Although this pragmatic approach sacrifices some of the broader generality found in~\cite{baggio2017lti,willems2012open}, sidestepping these technicalities allows us to formulate simple, tractable, and actionable methods for data-driven identification and control in subsequent sections.

\section{System identification and prediction} \label{sec:GB_sysID}
Having formally defined Gaussian behaviors, we now focus on how to identify them from finite data and leverage them for prediction.
Recall that a Gaussian behavior imposes no law on the nominal component $\overline{w}$.
However, to gather data from the system, we must fix the law of $\overline{w}$, which amounts to designing an experiment.
We take it to be a stationary zero-mean Gaussian (similarly to~\cite{willems2012open}) as summarized in the assumption below.

\begin{ass} \label{ass:endogenous_Gaussian}
    The available data consists of $D > qL$ length-$L$ trajectories 
    $w^\mathrm{d}_i = \overline{w}^\mathrm{d}_i + e_i$, 
    $i \in \{1, \dots, D\}$, of the Gaussian behavior $(\mathcal{B}_L, \Se)$, 
    where $\overline{w}^\mathrm{d}_i$ are identically 
    distributed zero-mean Gaussians with covariance $\overline{\Sigma}$, 
    independent of the exogenous noise $e_i \sim \mathcal{N}(0, \Se)$.
\end{ass}

In a typical control setup where the system is described by stochastic state-space equations, the experiment can be designed such that the data satisfy Assumption~\ref{ass:endogenous_Gaussian}.
Specifically, by starting from rest to ensure the initial state has zero mean and then applying stationary zero-mean Gaussian inputs, the linearity of the system guarantees that $\overline{w}^\mathrm{d}$ follows a stationary zero-mean Gaussian distribution.
Under Assumption~\ref{ass:endogenous_Gaussian}, the entire trajectory follows a zero-mean Gaussian, that is
\begin{align*}
    w^\mathrm{d}_i \sim \mathcal{N}(0,\Sfull), \quad \forall i \in\{1,\dots,D\},
\end{align*}
where $\Sfull = \overline{\Sigma} + \Se$.
The Gaussian behavior $(\mathcal{B}_L,\Se)$ is not identifiable from $\Sfull$ without further structural assumptions, as all decompositions inducing the same $\Sfull$ are observationally equivalent (see~\cite{willems2012open} for a thorough discussion on this non-uniqueness).
Specifically, any given covariance matrix can be split into a low-rank positive semi-definite matrix (endogenous randomness) and a residual positive semi-definite matrix (exogenous randomness) in infinitely many ways.
Since our objective is predictive control, we shift our focus from identifying the complete Gaussian behavior to deriving and estimating a predictive distribution.
This predictive distribution is uniquely identifiable from $\Sfull$, and gives rise to a possible Gaussian behavior consistent with the data.

\subsection{Estimation of the predictive distribution} \label{sec:pred}
For prediction tasks, part of the trajectory is known and the remainder is to be predicted.
Accordingly, we partition a trajectory into free variables $\wfree \in \mathbb{R}^{d_\mathrm{f}}$ and dependent variables $\wdep \in \mathbb{R}^{qL-d_\mathrm{f}}$.
In a conventional state-space setting, $\wfree$ corresponds to the initial conditions and the future input sequence, while $\wdep$ captures the future outputs.
The prediction is then described by the conditional probability of $\wdep$ given $\wfree$.
Since the underlying trajectories follow a joint Gaussian distribution, this conditional distribution is also Gaussian and takes the form
\begin{equation} \label{eq:pred_distr}
    \wdep~|~\wfree \sim \mathcal{N}(M\wfree,\Sp).
\end{equation}
Here, the mean is a linear function of $\wfree$ defined by the matrix $M$, and $\Sp$ is the residual covariance matrix.
The partitioning divides the randomness into an endogenous part, attributed to $\wfree$, and an exogenous part, captured by $\Sp$. 
Conditioning treats $\wfree$ as a given argument, so~\eqref{eq:pred_distr} describes the law of $\wdep$ as a function of $\wfree$, without specifying the law of $\wfree$.

Let us partition the covariance $\Sfull$ into terms corresponding to the free and dependent variables using a permutation matrix $\Pi$ as 
    \begin{align*}
        \Pi \Sfull \Pi^\top = \begin{bmatrix} \Sfull_\mathrm{ff} & (\Sfull_\mathrm{df})^\top \\ \Sfull_\mathrm{df} & \Sfull_\mathrm{dd} \end{bmatrix}.
    \end{align*}
\begin{ass} \label{ass:free_rank}
    Assume that $\Sfull_\mathrm{ff}$ has full rank.
\end{ass}
Assumption~\ref{ass:free_rank} requires the free variables to be sufficiently excited.
Suppose $\wfree$ forms a minimal parameterization of $\mathcal{B}_L$, that is, $d_\mathrm{f} = d$ and every $\overline{w} \in \mathcal{B}_L$ is uniquely determined by its free variables.
Then, $\mathrm{im}(\overline{\Sigma}) = \mathcal{B}_L$ implies that the nominal free block $\overline{\Sigma}_\mathrm{ff}$ is invertible, and hence $\Sfull_\mathrm{ff} \succeq \overline{\Sigma}_\mathrm{ff} \succ 0$.
Conversely, if $d_\mathrm{f} > d$, the condition $\Sfull_\mathrm{ff} \succ 0$ might still hold due to the exogenous noise in the data.
Under Assumption~\ref{ass:free_rank}, the distribution~\eqref{eq:pred_distr} can be uniquely identified from $\Sfull$ using standard formulas~\cite[Thm. 1.2.11]{muirhead1982aspects}
    \begin{align} \label{eq:conditioning}
    \begin{split}
        M & = \Sfull_\mathrm{df} (\Sfull_\mathrm{ff})^{-1}, \\
        \Sp & = \Sfull_\mathrm{dd} - \Sfull_\mathrm{df} (\Sfull_\mathrm{ff})^{-1} (\Sfull_\mathrm{df})^\top.
    \end{split}
    \end{align}
In practice, the joint covariance $\Sfull$ is unknown.
We therefore estimate it by the sample covariance of the data, which yields the following estimate of the predictive distribution.

\begin{prop} \label{prop:ID_pred}
Let $W^\mathrm{free}$ and $W^\mathrm{dep}$ denote the rows of the data matrix 
$\Wd$ from Section~\ref{sec:prelim-dd-control} corresponding to $\wfree$ and $\wdep$, respectively.
Let Assumption~\ref{ass:endogenous_Gaussian} hold and 
suppose that $\mathrm{rank}(W^\mathrm{free}) = d_\mathrm{f}$.
Then, setting $\Sfullh = \tfrac{1}{D} WW^\top$ and substituting it into~\eqref{eq:conditioning} 
yields
\begin{subequations} \label{eq:pred_both}
\begin{align} 
    \hat{M} & = W^\mathrm{dep}(W^\mathrm{free})^\dagger, \label{eq:pred_mean}\\ 
    \Sph & = \frac{1}{D} W^\mathrm{dep}\left(I-(W^\mathrm{free})^\dagger W^\mathrm{free}\right)(W^\mathrm{dep})^\top. \label{eq:pred_var}
\end{align}
\end{subequations}
\end{prop}

The rank condition on $W^\mathrm{free}$ is the sample counterpart of Assumption~\ref{ass:free_rank}, and it is reminiscent of the persistency of excitation condition for deterministic behaviors~\cite{markovsky2022data}.
Under the additional assumption that $w^d_i$ are independent, $\hat{M}$ and $\Sph$ in Proposition~\ref{prop:ID_pred} are the maximum likelihood estimates of $M$ and $\Sp$ in\eqref{eq:pred_distr}~\cite[Result 7.10]{johnson2007applied}.
Furthermore, the estimate $\hat{M}$ is the optimal linear predictor of $\wdep$, in the sense that the prediction error covariance $\Sph$ is minimal among the linear predictors~\cite[Thm. 3.2.2]{kailath2000linear}.

While the predictive distribution~\eqref{eq:pred_distr} is unique, it corresponds to one possible Gaussian behavior consistent with the data, which we now construct.
Let the trajectory be partitioned as $w = \Pi\begin{bmatrix}
    \wfree^\top & \wdep^\top
\end{bmatrix}^\top$.
Then, a possible nominal trajectory is given as $\overline{w} = \Pi \begin{bmatrix}
    \wfree^\top & (M\wfree)^\top
\end{bmatrix}^\top$, with noise $e = \Pi \begin{bmatrix}
    0 & e_\mathrm{pred}^\top
\end{bmatrix}^\top$.
The estimates of Proposition~\ref{prop:ID_pred} thus induce the Gaussian behavior $(\hat{\mathcal{B}}_L, \Seh)$, with
\begin{align*}
    \hat{\mathcal{B}}_L = \mathrm{im} \left( \Pi^\top \begin{bmatrix} I \\ \hat{M} \end{bmatrix} \right), \quad \text{and} \quad \Pi \Seh \Pi^\top  = \begin{bmatrix} 0 & 0 \\ 0 & \Sph \end{bmatrix}.
\end{align*}
If $\wfree \sim \mathcal{N}(0,(1/D)\cdot W^\mathrm{free}(W^\mathrm{free})^\top)$, this Gaussian behavior reproduces the sample covariance $\Sfullh$.
Note that $\hat{\mathcal{B}}_L$ has dimension $d_\mathrm{f}$, which is dictated by the 
chosen partitioning into $\wfree$ and $\wdep$.
The following remarks are in order.

\begin{rmk}
    We note that in case no exogenous noise is present ($\Se=0$), the sample covariance is a direct data-driven representation of the behavior, since $\mathrm{im}(\Sfullh) = \mathrm{im}(\overline{\Sigma}) \subseteq \mathcal{B}_L$.
    As a result, existing data-driven behavioral or subspace methods that rely on the image of $W$ can be reinterpreted within our framework by using the image of $\Sfullh$.
    In fact, with $\wfree = \begin{bmatrix}
        \wi^\top & \uf ^\top 
    \end{bmatrix}^\top$ and $\wdep = \yf$, the mean predictor in~\eqref{eq:pred_mean} equals the subspace predictor in~\eqref{eq:SPC}.
\end{rmk}

\begin{rmk}
    In the setting considered here, identifying the matrix $M$ in~\eqref{eq:pred_distr} becomes an \gls{EIV} problem, since the regressor $\wfree$ is itself corrupted by noise. 
    A broad literature exists on methods for \gls{EIV} system identification~\cite{soderstrom2007errors,soderstrom2018errors}, but it is well known that the system is generally not identifiable without introducing additional structural or statistical assumptions on the noise and/or latent variables.
    Moreover, the objective of \gls{EIV} methods is typically to recover the underlying “true” system model, which does not necessarily coincide with the model that yields the best predictive performance. 
    Since our primary goal is predictive control, we instead focus on identifying the predictive distribution~\eqref{eq:pred_distr} that provides the optimal predictor from the subspace system identification literature~\cite{katayama2005subspace}.
\end{rmk}

\subsection{Confidence bound on the prediction} \label{sec:sys_ID_uncertainty}
The covariance $\Sph$ identified in Proposition~\ref{prop:ID_pred} describes the \emph{aleatoric} uncertainty of the prediction, that is, the irreducible randomness caused by the exogenous noise, which persists even if the model were known exactly.
Additionally, the estimates $\hat{M}$ and $\Sph$ are themselves uncertain, since they are identified from a finite dataset.
This constitutes \emph{epistemic} uncertainty, which stems from limited data and vanishes as more data is collected.
Accounting for both sources, the predictions are guaranteed to be confined within a confidence ellipsoid~\cite[Sec. 7.7]{johnson2007applied}.
We summarize this result in the following lemma.

\begin{lemma} \label{lemma:confidence_ellipsoid}
    Let Assumption~\ref{ass:endogenous_Gaussian} hold. 
    Suppose that $w^\mathrm{d}_i$ are mutually independent, and $\Sph$ is invertible.
    Then, given $\wfree$ and a confidence level $0<\epsilon<0.5$, the true value $\wdep = M\wfree + e_\mathrm{pred}$ satisfies with probability $1-\epsilon$
    \small
    \begin{align} \label{eq:confidence_ellipsoid}
        \|\wdep - \hat{M}\wfree\|^2_{\Sph^{-1}} \leq \left(1 + \|\wfree\|^2_{(W^\mathrm{free} (W^\mathrm{free})^\top)^{-1}} \right) \cdot r(\epsilon),
    \end{align}
    \normalsize
    where $r(\epsilon):= \frac{D\cdot (qL-d)}{D-qL +1}F_{d,D-qL +1}(1-\epsilon)$, and $F_{d,D-qL +1}(1-\epsilon)$ is the inverse F-distribution at probability $1-\epsilon$ with degrees of freedom $d$ and $D-qL+1$.
\end{lemma}

The bound in Lemma~\ref{lemma:confidence_ellipsoid} quantifies the prediction uncertainty as a function of the available data.
As the number of data samples $D$ tends to infinity, the matrix $W^\mathrm{free}(W^\mathrm{free})^\top$ grows unbounded, causing its inverse to approach zero.
As a result, the term $\|\wfree\|_{(W^\mathrm{free}(W^\mathrm{free})^\top)^{-1}}$ vanishes, and the right-hand side of~\eqref{eq:confidence_ellipsoid} converges to $r(\epsilon)$. 
This asymptotic behavior corresponds to the case where the parameters $\hat{M}$ and $\Sph$ are known exactly, meaning the epistemic uncertainty disappears and the prediction uncertainty is determined solely by the aleatoric part.
For finite data sets, the epistemic term $\|\wfree\|_{(W^\mathrm{free}(W^\mathrm{free})^\top)^{-1}}$ is nonzero and makes the size of the confidence ellipsoid direction-dependent.
Directions in the free-variable space where more data has been collected correspond to larger eigenvalues of $W^\mathrm{free}(W^\mathrm{free})^\top$.
Consequently, when $\wfree$ points in these data-rich directions, the value of $\|\wfree\|_{(W^\mathrm{free}(W^\mathrm{free})^\top)^{-1}}$ is small, resulting in a smaller confidence ellipsoid.

We also note that Lemma~\ref{lemma:confidence_ellipsoid} relies on the assumption that the data samples $w_i^d$ are mutually independent.
This condition requires each column of the data matrix $W$ to be collected from a separate, independent experiment. 
In practice, constructing a Hankel matrix from a single continuous data stream is a more data-efficient approach.
Relaxing the independence assumption to accommodate such structured matrices or informed experiment design are left as future work.

\subsection{Incorporating uncertain forecasts in the prediction} \label{sec:forecast}
Predictive control benefits greatly from forecasts of future quantities, making the treatment of forecast uncertainty an important consideration.
For instance, forecasts of weather, renewable generation, or consumer demand are central to the economic operation of energy systems and buildings.
Our framework uses a multi-step model, which is better suited than single-step alternatives for incorporating such forecasts over a horizon.
As an illustrative example, consider a stochastic state-space representation~\eqref{eq:stochastic-SS} that is augmented with an external variable $v_t$  affecting the state dynamics, such as $$x_{t+1} = Ax_t + Bu_t + v_t + \xi_t.$$
Assume that at each time $t'$, the past values $v_t$ for $t \le t'$ are observed, while future values are available only through forecasts.
We write $\hat{v}_{t|t'}$ for the forecast of $v_t$ issued at time $t'$.
The forecast for the entire future horizon predicted at time $t'$ is $\hat{v}_{\mathrm{f}|t'} = \begin{bmatrix}
    \hat{v}_{t'+1|t'}^\top & \dots & \hat{v}_{t'+\Tf|t'}^\top
\end{bmatrix}^\top$, and it is a sum of the true future values $v_\mathrm{f}$ and a forecast noise component $e_v$, such that $\hat{v}_{\mathrm{f}|t'} = v_\mathrm{f} + e_v$.
Our goal is to synthesize a predictive distribution of $\yf$ given $\ui,\yi,\uf$, the most recent measurements of $v$ denoted as $v_\mathrm{ini}$, and the forecast $\hat{v}_{\mathrm{f}|t'}$.

In order to obtain such a predictive distribution, we define the ``proxy" trajectory for some time $t$ as
$$\begin{bmatrix}
    u_{[t-\Ti,t]} \\  v_{[t-\Ti,t]} \\ y_{[t-\Ti,t]} \\ u_{[t+1,t+\Tf]} \\ \hat{v}_{\mathrm{f}|t} \\ y_{[t+1,t+\Tf]}
\end{bmatrix} = \begin{bmatrix}
    \ui \\ v_\mathrm{ini} \\ \yi \\ \uf \\ v_\mathrm{f} \\ \yf 
\end{bmatrix} + \begin{bmatrix}
    0 \\ 0 \\ 0\\ 0\\ e_v \\ 0
\end{bmatrix}.$$
The noise term originating from the forecast uncertainty $e_v$ can be simply lumped to the exogenous noise $e$ in Definition~\ref{def:GB}.
Thus, we can use the data-driven prediction method described in Section~\ref{sec:pred}.
The dependent variable corresponds to $\yf$, and we include all other components in $\wfree$.
The corresponding partitioned data matrix is $W^\mathrm{free} = \begin{bmatrix}
    U_\mathrm{p}^\top & V_\mathrm{p}^\top & Y_\mathrm{p}^\top & U_\mathrm{f}^\top & F^\top
\end{bmatrix}^\top,$
where $V_\mathrm{p}$, containing the initial values of $v$, is defined similarly to $U_\mathrm{p}$.
Furthermore, $F$ contains $D$ separate forecasts corresponding to the offline data arranged as
$$
F = \begin{bmatrix}
    \hat{v}_{\mathrm{f}|1} & \hat{v}_{\mathrm{f}|2} & \dots & \hat{v}_{\mathrm{f}|D}
\end{bmatrix}.
$$
With the introduced proxy trajectory, the desired predictive distribution is constructed as in~\eqref{eq:pred_distr}.
The components $M$ and $\Sp$ can be estimated with the augmented data matrix $W^\mathrm{free}$ as in Proposition~\ref{prop:ID_pred}.

Importantly, the estimated covariance $\Sph$ now captures the forecast uncertainty as well. 
We impose no assumptions on the temporal correlations of the forecast error $e_v$, which may be arbitrarily correlated across the horizon, and it need not be stationary.
For instance, the forecast uncertainty may grow over the prediction horizon, as is typical of real forecasts.
Consequently, the covariance $\Sp$ does not possess any particular structure.
By including past forecasts in the matrix $W$, the proposed method identifies this unstructured uncertainty directly from data.
This flexibility stems from the trajectory-based approach, which treats the horizon jointly and requires no temporal recursion on the errors.
In fact, the proxy trajectory does not come from a window of an infinite sequence, because the forecast $\hat{v}_{\mathrm{f}|t}$ may change entirely between two consecutive times $t$ and $t+1$. 
Nevertheless, our framework can still handle this setting, since the estimation in Proposition~\ref{prop:ID_pred} relies only on the joint distribution of the trajectory samples.
Single-step models, by contrast, propagate errors recursively, which imposes structural assumptions on the disturbances and makes such time-varying forecast uncertainty more challenging to model.
We illustrate the treatment of noisy forecasts on the \gls{HVAC} case study in Section~\ref{sec:HVAC}.

\section{Control under aleatoric uncertainty} \label{sec:CE_control}
We now apply our framework to the design of data-driven predictive control. 
To this end, we utilize the predictive distribution given in~\eqref{eq:pred_distr} with the assignments $\wfree = \begin{bmatrix} \wi^\top & \uf^\top \end{bmatrix}^\top$ and $\wdep = \yf$.
By partitioning the data matrix $W$ as detailed in Section~\ref{sec:prelim-dd-control}, the parameters of the predictive distribution, $\hat{M}$ and $\Sph$, are estimated from data according to~\eqref{eq:pred_both} in Proposition~\ref{prop:ID_pred}, yielding
\begin{align} \label{eq:pred_distr_control_estim}
\hat{M} = Y_\mathrm{f} \begin{bmatrix} W_\mathrm{p} \\ U_\mathrm{f} \end{bmatrix}^\dagger,\; \Sph = \frac{1}{D} Y_\mathrm{f}\left(I-\begin{bmatrix} W_\mathrm{p} \\ U_\mathrm{f} \end{bmatrix}^\dagger\begin{bmatrix} W_\mathrm{p} \\ U_\mathrm{f} \end{bmatrix}\right)Y_\mathrm{f}^\top.
\end{align}
This matrix formulation is widely recognized in the data-driven control literature as the \textit{subspace predictor} in~
\eqref{eq:SPC}~\cite{favoreel1999spc,markovsky2022data}.
Note also that this predictive distribution inherently corresponds to the optimal prediction step utilized in subspace system identification, see~\cite{katayama2005subspace}.
However, rather than proceeding with the traditional realization step to explicitly identify state-space system matrices ($A, B, C, D$), we bypass state estimation entirely and use this data-driven predictive distribution directly.

In this section, we only account for the aleatoric uncertainty, i.e., we assume that $\hat{M}$ and $\Sph$ are perfect estimates of $M$ and $\Sp$.
We express the predictive distribution~\eqref{eq:pred_distr} as an affine relationship with additive noise
\begin{align} \label{eq:pred_distr_control}
\yf = \hat{M}_\mathrm{ini} \wi + \hat{M}_u \uf + e_\mathrm{pred},
\end{align}
where $\hat{M}_\mathrm{ini}$ and $\hat{M}_u$ denote the block columns of $\hat{M}$ corresponding to $\wi$ and $\uf$, respectively.
The term $e_\mathrm{pred} \sim \mathcal{N}(0, \Sph)$ represents the aleatoric uncertainty in the form of stochastic prediction error.
Building on the relation~\eqref{eq:pred_distr_control}, we design predictive controllers that minimize the expected cost subject to (probabilistic) constraints on the future inputs and outputs. 
We present two formulations: first with a feedforward input, which recovers \gls{SPC}, and then with a disturbance affine feedback policy.

\subsection{Chance-constrained stochastic predictive control} \label{sec:CE_feedforward}
We consider polytopic constraint sets in the form
\begin{align*}
    \mathcal{Y} &:= \left\{y\in\mathbb{R}^{p\Tf}~|~a_i^\top y\leq b_i,~\forall i= 1,\dots,n_y \right\}, \\ 
    \mathcal{U} &:= \left\{u\in\mathbb{R}^{m\Tf}~|~(a^u_i)^\top u \leq b^u_i,~\forall i= 1,\dots,n_u \right\}, 
\end{align*}
defined by $n_y$ and $n_u$ scalar inequalities, respectively.
Since the future outputs (and potentially the inputs) are random variables, as standard in stochastic \gls{MPC}~\cite{farina2016stochastic,mesbah2016stochastic}, we enforce individual chance constraints in the form $\mathrm{Pr}(a_i^\top y \leq b_i) \geq 1-\epsilon$, for each constraint $i \in \{1, \dots, n_y\}$, where $\epsilon \in (0, 0.5]$ is a predefined acceptable violation probability.
For Gaussian random variables, individual chance constraints admit an equivalent deterministic representation using the mean and variance~\cite{blackmore2011chance}.
\begin{lemma} \label{lemma:chance_constraint}
For $x\sim\mathcal{N}(\mu,\sigma^2)$, it holds that
    \begin{align} \label{eq:lemma_chance_constr1}
        \mathrm{Pr}(x\leq \beta)\geq 1-\epsilon,
        \end{align}
        with some $\epsilon\in (0,0.5]$, if and only if 
        \begin{align}  \label{eq:chance_constraint_det}
        \mu + c(\epsilon) \sigma \leq \beta,
    \end{align}
    where $$c(\epsilon):= \sqrt{2}\mathrm{erf}^{-1}(1-2\epsilon),$$ and $\mathrm{erf}$ is the Gauss error function, satisfying $\mathrm{erf}^{-1}(0) = 0$.
\end{lemma}

\begin{rmk}
Alternative to chance constraints, one can employ the \gls{CVaR}~\cite{rockafellar2000optimization}, which accounts for the expected magnitude of constraint violations.
For a scalar Gaussian variable $x\in\mathbb{R},~x\sim\mathcal{N}(\mu,\sigma^2)$, the \gls{CVaR} takes the form~\cite{rockafellar2000optimization}
\begin{equation*}
\mathrm{CVaR}_\epsilon(x) = \mu + c_{\mathrm{CVaR}}(\epsilon)\sigma,
\end{equation*}
where
$
c_{\mathrm{CVaR}}(\epsilon) := \left(\sqrt{2\pi} \epsilon \exp\left(\mathrm{erf}^{-1}(1- 2\epsilon)\right)^2 \right)^{-1}.
$
Since enforcing $\mathrm{CVaR}_\epsilon(x) \leq 0$ leads to the same structural constraint as~\eqref{eq:chance_constraint_det}, with a different weighting on the standard deviation, we can incorporate \gls{CVaR} constraints into the subsequent control formulations by swapping the weighting.
\end{rmk}

Here, we treat the future input $\uf\in\mathcal{U}$ as a deterministic decision variable, while the output $\yf$ remains stochastic through the prediction error $e_\mathrm{pred}$.
Minimizing the expected control cost given $\wi$ subject to output chance constraints yields the following formulation
\begin{align} \label{eq:stoch_control}
\begin{split}
    \min_{\uf\in\mathcal{U}} & \quad \mathbb{E}[J(\uf,\yf)] \\
    \mathrm{s.t.} & \quad \mathrm{Pr}(a_i^\top \yf \leq b_i) \geq 1-\epsilon,\quad \forall i=1,\dots,n_y.
\end{split}
\end{align}
By expanding the expectation and applying Lemma~\ref{lemma:chance_constraint}, the stochastic problem~\eqref{eq:stoch_control} can be reformulated as the deterministic quadratic programming problem
\begin{align*}
    \min_{\uf\in\mathcal{U}} & \quad J\left(\uf,\hat{M}_\mathrm{ini} \wi + \hat{M}_u \uf \right) + \mathrm{tr}(Q\Sph) \\
\mathrm{s.t.} & \quad a_i^\top (\hat{M}_\mathrm{ini} \wi + \hat{M}_u \uf)
 + c(\epsilon) \sqrt{a_i^\top \Sph a_i} \leq b_i.
\end{align*}
We recover \gls{SPC}~\eqref{eq:SPC} by setting $\epsilon = 0.5$.
The mean output $\mathbb{E}[\yf]=\hat{M}_\mathrm{ini}\wi+\hat{M}_\mathrm{u}\uf$ equals the subspace predictor by~\eqref{eq:pred_distr_control_estim}.
Since $c(0.5)=0$, the constraints act on this mean, and the trace term $\mathrm{tr}(Q\Sph)$ is a constant offset independent of $\uf$.
Therefore, the two problems yield the same optimal input.

\subsection{Predictive control with affine feedback policies} \label{sec:dist_feedback}

While the formulation~\eqref{eq:stoch_control} provides a useful baseline, restricting the future input sequence $\uf$ to be purely deterministic inherently leads to conservative control actions.
Under open-loop predictions, the uncertainty of the future outputs naturally accumulates over the horizon due to the persistent exogenous noise $e_\mathrm{pred}$. 
Guaranteeing constraint satisfaction under these expanding bounds restricts the feasible region and degrades performance.

However, because $\yf$ is a trajectory of an underlying dynamical system, the elements of the prediction error $e_\mathrm{pred}$ are temporally correlated. 
This correlation is explicitly captured by the estimated covariance matrix $\Sph$.
We can exploit this knowledge to actively mitigate the growth of prediction uncertainty by updating future inputs based on intermediate error realizations.
Since the mean prediction is linear and the noise is Gaussian, it is natural to replace the open-loop sequence $\uf$ with a disturbance affine feedback policy
\begin{align} \label{eq:feeback_policy}
    \uf = K e_\mathrm{pred} + \overline{u}_\mathrm{f}
\end{align}
where $\overline{u}_\mathrm{f}$ denotes the deterministic feedforward component of the input, and $K$ is a causal feedback gain matrix.
Causality ensures that the control action at any future time step depends only on the prediction errors realized at previous time steps, which can be computed by comparing the realized and predicted outputs.
Mathematically, this restricts $K$ to the set of strictly block lower triangular matrices, defined as
\begin{align*}
\mathcal{K} := \left\{ K \in \mathbb{R}^{m\Tf \times p\Tf} ~|~ K_{i,j} = 0,~\forall i \leq j \right\},
\end{align*}
where $K_{i,j}$ denotes the $(i,j)$-th block of dimension $m \times p$.
The feedback term $K e_\mathrm{pred}$ in~\eqref{eq:feeback_policy} incorporates the exogenous noise into the system input, rendering the sequence $\uf$ stochastic. 
Optimizing the feedback gain $K\in\mathcal{K}$ then allows us to actively shape the covariance of the predicted output $\yf$, reducing its variance and mitigating the conservatism of the predictions, at the expense of increasing the variance of the applied control input

We optimize over the feedforward input $\overline{u}_\mathrm{f}$ and the causal gain $K\in\mathcal{K}$, which parameterize the policy~\eqref{eq:feeback_policy}.
Under this policy, $\uf$ and $\yf$ are the random variables defined by~\eqref{eq:feeback_policy} and~\eqref{eq:pred_distr_control}.
The resulting chance constrained predictive control problem is
\begin{subequations} \label{prb:feedback_policy}
\begin{align}
\min_{\overline{u}_\mathrm{f},K\in\mathcal{K}} & \; \mathbb{E}[J(\uf,\yf)] \\
\mathrm{s.t.} &  \; \mathrm{Pr}(a_i^\top\yf \leq b_i) \geq 1- \epsilon_y,\; \forall i = 1,\dots,n_y, \label{eq:chance_constraints_y}\\
    & \; \mathrm{Pr}((a^u_j)^\top \uf \leq b^u_j) \geq 1- \epsilon_u,\; \forall j = 1,\dots,n_u. \label{eq:chance_constraints_u}
\end{align}
\end{subequations}
Since the inputs are now stochastic, we pose individual chance constraints on the inputs as well in~\eqref{eq:chance_constraints_u}.
The input and output constraints must be satisfied with probability $1-\epsilon_u$ and $1-\epsilon_y$, respectively, where $\epsilon_u, \epsilon_y\in(0,0.5]$.
Exploiting that $\uf$ and $\yf$ are Gaussians and the constraints are polytopic, problem~\eqref{prb:feedback_policy} can be reformulated as the following deterministic optimization problem.

\begin{thm}
Problem~\eqref{prb:feedback_policy} is equivalent to the following convex problem
\begin{align} \label{eq:cvx_reform}
\begin{split}
    \min_{\overline{u}_\mathrm{f},K\in\mathcal{K}} \quad & J\left(\overline{u}_\mathrm{f}, \hat{M}\begin{bmatrix}
        \wi \\ \overline{u}_\mathrm{f}
    \end{bmatrix}\right) +  \mathrm{tr}\big(R\cdot K\Sph K^\top\big)\\
        & \qquad + \mathrm{tr}\left(Q\cdot(\hat{M}_\mathrm{u} K + I)\Sph (\hat{M}_\mathrm{u} K +I)^\top\right) \\
        \mathrm{s.t.} \quad 
        & c(\epsilon_y)\cdot\|(\hat{M}_u K+I)^\top a_i\|_{\Sph} \leq  b_i - a_i^\top \hat{M}\begin{bmatrix}
            \wi \\ \overline{u}_\mathrm{f}
        \end{bmatrix},\\
        & c(\epsilon_u)\cdot \|K^\top a^u_j\|_{\Sph} \leq b^u_j - (a^u_j)^\top \overline{u}_\mathrm{f}, \\
        & \forall i = 1,\dots,n_y, \quad \forall j = 1,\dots,n_u.
\end{split}
\end{align}
\end{thm}

\begin{proof}
The affine feedback policy~\eqref{eq:feeback_policy} renders the input stochastic, with distribution
\begin{align}
    \label{eq:uf_distr} 
    \uf \sim \mathcal{N}(\overline{u}_\mathrm{f},K\Sph K^\top).
\end{align}
Further,~\eqref{eq:pred_distr_control} with the feedback policy is
\begin{align*}
    \yf = \hat{M}_\mathrm{ini} \wi + \hat{M}_\mathrm{u} \uf + (\hat{M}_\mathrm{u}K +I)e_\mathrm{pred}.
\end{align*}
Thus, the distribution of $\yf$ becomes
\begin{align}
    \yf \sim \mathcal{N}\Bigg(\hat{M}\begin{bmatrix}
        \wi \\ \overline{u}_\mathrm{f}
    \end{bmatrix}, (\hat{M}_\mathrm{u} K +I)\Sph(\hat{M}_\mathrm{u} K +I)^\top\Bigg). \label{eq:yf_distr}
\end{align}
Note that, for $x\sim \mathcal{N}(\mu,\Sigma)$, the expectation of a quadratic cost decomposes as $\mathbb{E}[x^\top Qx] = \mu^\top Q\mu + \mathrm{tr}(Q\Sigma)$.
The equivalence of the costs in~\eqref{prb:feedback_policy} and \eqref{eq:cvx_reform} can be shown by inserting~\eqref{eq:uf_distr} and~\eqref{eq:yf_distr} into $\mathbb{E}[J(\uf,\yf)]$.
Furthermore, the cost is quadratic in both $\overline{u}_\mathrm{f}$ and $K$.

Moreover with $\uf$ and $\yf$ distributed as in~\eqref{eq:uf_distr} and~\eqref{eq:yf_distr}, the constraints of~\eqref{prb:feedback_policy} and \eqref{eq:cvx_reform} are equivalent by Lemma~\ref{lemma:chance_constraint}.
Note that the constraints of~\eqref{eq:cvx_reform} are second-order cone constraints, and $K\in\mathcal{K}$ is defined by linear equality constraints.
Thus,~\eqref{eq:cvx_reform} is a convex problem.
\end{proof}

The first term in~\eqref{eq:cvx_reform}, involving the feedforward input $\overline{u}_\mathrm{f}$, penalizes the mean cost, while the two trace terms, involving the feedback gain $K$, penalize the variance and take the form of an LQG cost in disturbance feedback form.
In Section~\ref{sec:HVAC}, we demonstrate this formulation and show that optimizing over disturbance affine feedback policies reduces the controller's conservatism.

\section{Control under aleatoric and epistemic uncertainty} \label{sec:uncertain_control}
As we established in Section~\ref{sec:sys_ID_uncertainty}, identifying the parameters $M$ and $\Sp$ from finite data as in Proposition~\ref{prop:ID_pred} leads to epistemic uncertainty.
In this section, we propose predictive control methods that take into account this uncertainty.
Unlike the aleatoric noise, the epistemic uncertainty (and thus, the overall uncertainty in the predictions) depends on the input.
Consequently, optimizing over disturbance feedback policies as in Section~\ref{sec:dist_feedback} would lead to non-convex problems.
We therefore restrict our attention to feedforward inputs, and design controllers using Lemma~\ref{lemma:confidence_ellipsoid}.
Assuming that $\Sph$ is invertible, we define the uncertainty ellipsoid around the predicted $\yf$ as
\begin{align}
\begin{split} \label{eq:def_M}
& \mathcal{M}_\epsilon(\uf,\wi) := \Biggl\{y\in\mathbb{R}^{p\Tf}~\Bigg| \\ 
&  \left\|y-\hat{M}\begin{bmatrix}
    \wi \\ \uf
\end{bmatrix}\right\|^2_{\Sph^{-1}} \leq \left(1 + \left\|\begin{bmatrix}
    \wi \\ \uf 
\end{bmatrix}\right\|^2_{P}\right)\cdot r(\epsilon)
\Biggr\},
\end{split}
\end{align}
where $$P:= \left(\begin{bmatrix}
    W_\mathrm{p} \\U_\mathrm{f}
\end{bmatrix}\begin{bmatrix}
    W_\mathrm{p} \\U_\mathrm{f}
\end{bmatrix}^\top\right)^{-1},$$ $r(\epsilon)$ is defined as in Lemma~\ref{lemma:confidence_ellipsoid} with $d =mL + p\Ti$, and $\epsilon$ assesses the risk.
By Lemma~\ref{lemma:confidence_ellipsoid}, we know that the ellipsoid $\mathcal{M}_\epsilon(\uf,\wi)$ contains the true value of $\yf$ with probability $1-\epsilon$.
Depending on the designer's attitude toward the uncertainty captured by $\mathcal{M}_\epsilon$, two distinct control paradigms emerge.

\subsection{Robust control} \label{sec:robust}
In many control applications, minimizing the risk of poor performance due to model errors is an important requirement~\cite{mesbah2016stochastic}.
This naturally motivates a robust control framework where, we safeguard against the worst-case prediction within the confidence ellipsoid $\mathcal{M}_\epsilon(\uf,\wi)$.
This perspective leads to the min-max optimization problem
\begin{align} \label{prb:robust_control}
    \min_{\uf\in\mathcal{U}} ~\max_{\yf \in\mathcal{M}_\epsilon(\uf,\wi)}  J(\uf,\yf).
\end{align}
Problem \eqref{prb:robust_control} ensures that the computed control input $\uf$ remains performant even under the most adversarial true future output sequence within the ambiguity set around the estimate. 
While min-max optimization problems are generally computationally challenging, the specific structure of the predictive distribution~\eqref{eq:pred_distr_control} and the Gaussianity assumption allow for a tractable reformulation, as detailed in the following theorem.
\begin{thm} \label{thm:robust}
    The inner maximization problem in~\eqref{prb:robust_control} is upper bounded by
    \begin{align} \label{eq:robust_bound}
    \begin{split}
     \|\uf - & \ur\|_R^2 + \left\|\hat{M}\begin{bmatrix}
        \wi \\ \uf
    \end{bmatrix} - \yr \right\|_{\left( Q^{-1} - \frac{1}{\lambda} \Sph \right)^{-1}}^2 \\
    & +  \lambda r(\epsilon)\left \| \begin{bmatrix}
        \wi \\ \uf
    \end{bmatrix}\right \|_{P}^2 + \lambda r(\epsilon)
    \end{split}
    \end{align}
    for any $\lambda > \|Q\Sph\|_2$ and any $\uf$.
\end{thm}
\begin{proof}
First, we rewrite the inner problem in~\eqref{prb:robust_control} as the minimization
\begin{align} \label{eq:proof_dual}
\min_{\yf \in\mathcal{M}_\epsilon(\uf,\wi)} \quad & -\|\uf - \ur\|_R^2 - \|\yf - \yr\|_Q^2.
\end{align}
The Lagrange dual function of~\eqref{eq:proof_dual} with multiplier $\lambda \ge 0$ is given by the unconstrained minimization problem
\begin{align*}
f(\lambda) =  \min_{\yf} & - \|\yf - \yr\|_Q^2 + \lambda \left \|\yf - \hat{M}\begin{bmatrix}
        \wi \\ \uf
    \end{bmatrix}\right\|^2_{\Sph^{-1}} \\
& - \|\uf-\ur\|_R^2 
- \lambda r(\epsilon)\left\|\begin{bmatrix}
    \wi \\ \uf
\end{bmatrix}\right\|_{P}^2 - \lambda r(\epsilon).
\end{align*}
The cost in $f(\lambda)$ is quadratic and strictly convex in $\yf$, since $\lambda>\|Q\Sph\|_2$ implies $\lambda \Sph^{-1} - Q \succ 0$.
The minimizer in $f(\lambda)$ is the average of $\yr$ and $\hat{M}\begin{bmatrix}
        \wi \\ \uf
    \end{bmatrix}$
weighted by $-Q$ and $\lambda\Sph^{-1}$~\cite[Chp. 15]{boyd2018introduction} and the negative dual function $-f(\lambda)$ coincides with the bound~\eqref{eq:robust_bound}.
Since the dual function is a lower bound on~\eqref{eq:proof_dual}, $-f(\lambda)$ is an upper bound on the inner maximization~\eqref{prb:robust_control} for any $\lambda$ and $\uf$.
\end{proof}

The upper bound established in Theorem~\ref{thm:robust} corresponds to the Lagrange dual of the inner maximization problem in~\eqref{prb:robust_control}, where $\lambda$ serves as the Lagrange multiplier associated with the uncertainty constraint.
The first term in~\eqref{eq:robust_bound} is the standard input penalty originating from the nominal control cost.
The second term minimizes the output tracking cost, while systematically incorporating the prediction uncertainty of $\yf$ as captured by the estimated covariance matrix $\Sph$.
The third term scales with the confidence level $\epsilon$ and accounts for the uncertainty introduced by the control actions.
As established in Lemma~\ref{lemma:confidence_ellipsoid}, the size of the confidence ellipsoid depends on the magnitude of the free variables; consequently, this term penalizes the application of large inputs $\uf$.
The fourth term is a constant with respect to $\uf$.

Because this bound is derived from the dual function, a joint minimization over both the control input $\uf$ and the multiplier $\lambda$ would yield an exact reformulation of the original min-max problem~\eqref{prb:robust_control}~\cite[Apx.~B]{boyd2004convex}.
However, this joint optimization results in a non-convex problem. 
Conversely, when $\lambda$ is fixed to any value satisfying $\lambda > \|Q\Sph\|_2$, \eqref{eq:robust_bound} is strictly convex with respect to $\uf$.
Being a dual function, \eqref{eq:robust_bound} is also convex in $\lambda$ for a fixed $\uf$.
One can therefore employ a primal-dual algorithm that alternates between these two convex problems.
For simplicity, we instead treat $\lambda$ as a constant tuning parameter rather than an optimization variable.
The choice of $\lambda$ dictates the trade-off between tracking performance and robustness to uncertainty.
Selecting large values for~$\lambda$ places the emphasis on the uncertainty constraints, leading to overly conservative control actions. 
On the other hand, as~$\lambda$ approaches the lower bound of $\|Q\Sph\|_2$, the second term in the cost function begins to dominate. 
In this case, the controller becomes less conservative and attempts to match the predicted mean closely to the reference $\yr$.

With $\lambda$ fixed to a chosen value, we can minimize the upper bound~\eqref{eq:robust_bound} with respect to $\uf$.
This leads to the following tractable convex problem that provides a conservative upper bound for the original min-max formulation in~\eqref{prb:robust_control}

\begin{align} \label{prb:robust}
    \min_{\uf\in\mathcal{U}}\quad  \eqref{eq:robust_bound}.
\end{align}

\begin{rmk}
From a risk-assessment perspective, the inner maximization problem in~\eqref{prb:robust_control} upper-bounds the \gls{VaR} of the control cost, defined as the infimum over all thresholds $\gamma$ satisfying $\mathrm{Pr}(J(\uf, \yf) \leq \gamma) \geq 1-\epsilon$.
Let $$\gamma^\star = \max_{\yf \in \mathcal{M}_\epsilon(\uf,\wi)} J(\uf, \yf)$$ denote the worst-case cost over the confidence ellipsoid.
Because $\gamma^\star$ is the maximum value, it holds that
$$\mathcal{M}_\epsilon(\uf,\wi) \subseteq \{\yf~|~ J(\uf, \yf) \le \gamma^\star\}.$$ 
Furthermore, the monotonicity of probability measures implies 
$$
\mathrm{Pr}\big(J(\uf, \yf) \le \gamma^*\big) \ge \mathrm{Pr}\big(\yf \in \mathcal{M}_\epsilon(\uf,\wi)\big) = 1-\epsilon.
$$
Thus, $\gamma^\star$ satisfies the probability threshold constraint, making it a valid upper bound on the \gls{VaR}, with respect to both the aleatoric and epistemic uncertainty.
\end{rmk}

\subsection{Optimistic control}
Predictive uncertainty can also be handled with ``optimism in the face of uncertainty", which is often employed in adaptive settings to foster exploration~\cite{abbasi2011regret}.
This perspective is of interest because it reveals that regularized \gls{DeePC} is fundamentally rooted in an optimistic interpretation of prediction uncertainty.
By optimizing over the best-case realizations of the future trajectory within the uncertainty set to reduce the control cost, we obtain the optimistic formulation
\begin{align*}
    \min_{\uf\in\mathcal{U}} ~\min_{\yf \in\mathcal{M}_\epsilon(\uf,\wi)} J(\uf,\yf)
\end{align*}

While this formulation captures the optimistic perspective, the minimization over $\uf$ renders the problem non-convex.
This non-convexity arises due to the epistemic uncertainty, which grows with the magnitude of the control input.
Specifically, larger inputs increase the prediction uncertainty, which expands the confidence ellipsoid.
The optimistic controller can exploit this expansion by selecting larger inputs to widen the bounds and achieve a lower control cost.
A practical way to prevent this behavior is to assume the uncertainty radius in~\eqref{eq:def_M} does not depend on the control effort (which is the case when there is only aleatoric uncertainty).
Replacing the right hand side of the bound~\eqref{eq:def_M} with a constant value $\alpha > 0$ yields the following tractable problem

\begin{align*}
\min_{\uf \in \mathcal{U}} \min_{\yf} & \quad J(u_f, y_f) \\
\text{s.t.} & \quad \left\|y_f - \hat{M}\begin{bmatrix}
    \wi \\ \uf
\end{bmatrix}\right\|_{\Sph^{-1}}^2 \leq \alpha.
\end{align*}
By applying Lagrange multipliers to soften the constraint on $\yf$, the min-min problem reduces to
\begin{align} \label{eq:min-min}
\min_{\uf \in \mathcal{U},\yf} \; J(u_f, y_f) + \lambda\left(\left\|y_f - \hat{M}\begin{bmatrix}
    \wi \\ \uf
\end{bmatrix}\right\|_{\Sph^{-1}}^2 - \alpha\right).
\end{align}

As established in our preliminary work~\cite{sasfi2025gaussian}, problem~\eqref{eq:min-min} recovers the regularized \gls{DeePC} algorithm~\eqref{eq:DeePC}.
\begin{thm} \label{thm:DeePC_eqv}
If $(\uf^\star,\yf^\star,g^\star)$ is a minimizer of \eqref{eq:DeePC} with $\mathcal{Y}=\mathbb{R}^{p\Ti}$, and $h(g)$ defined as in~\eqref{eq:proj-based-reg}, then $(\uf^\star,\yf^\star)$ is a minimizer of \eqref{eq:min-min} with $\lambda = \lambda_g \frac{2}{D}$.
\end{thm}
Theorem~\ref{thm:DeePC_eqv} establishes an exact correspondence between the optimistic formulation~\eqref{eq:min-min} and regularized \gls{DeePC}, revealing that \gls{DeePC} adopts an optimistic stance toward prediction uncertainty.
The two problems yield the same optimal input, and the Lagrange multiplier $\lambda$ plays the same role as the regularization weight $\lambda_g$ in~\eqref{eq:DeePC}.
The two problems differ only in their decision variables, as \gls{DeePC} optimizes over $g$, whose dimension scales with the data size $D$, whereas~\eqref{eq:min-min} optimizes directly over $\uf$ and $\yf$.
Moreover, since the uncertainty radius was fixed the a constant $\alpha$, \gls{DeePC} accounts only for the aleatoric noise and ignores epistemic uncertainty.

Although this optimistic approach provides a novel interpretation of \gls{DeePC} in the stochastic setting, it can be problematic in predictive control.
While the robust formulation in Section~\ref{sec:robust} guards against the worst-case output in the confidence ellipsoid, \eqref{eq:min-min} optimizes over the best case, using the freedom in $\yf$ to lower the control cost.
In particular, this optimism can severely degrade performance, especially when the confidence ellipsoid is large, or correspondingly, $\lambda$ is small.
This structural insight explains empirical observations in regularized \gls{DeePC}, where closed-loop performance has been shown to increase monotonically with the regularization parameter $\lambda_g$~\cite{dorfler2022bridging}.

\section{Case studies} \label{sec:case_studies}
\begin{figure*}[t!]
    \centering
    \begin{subfigure}[t]{0.49\textwidth}
        \includegraphics[width=\linewidth]{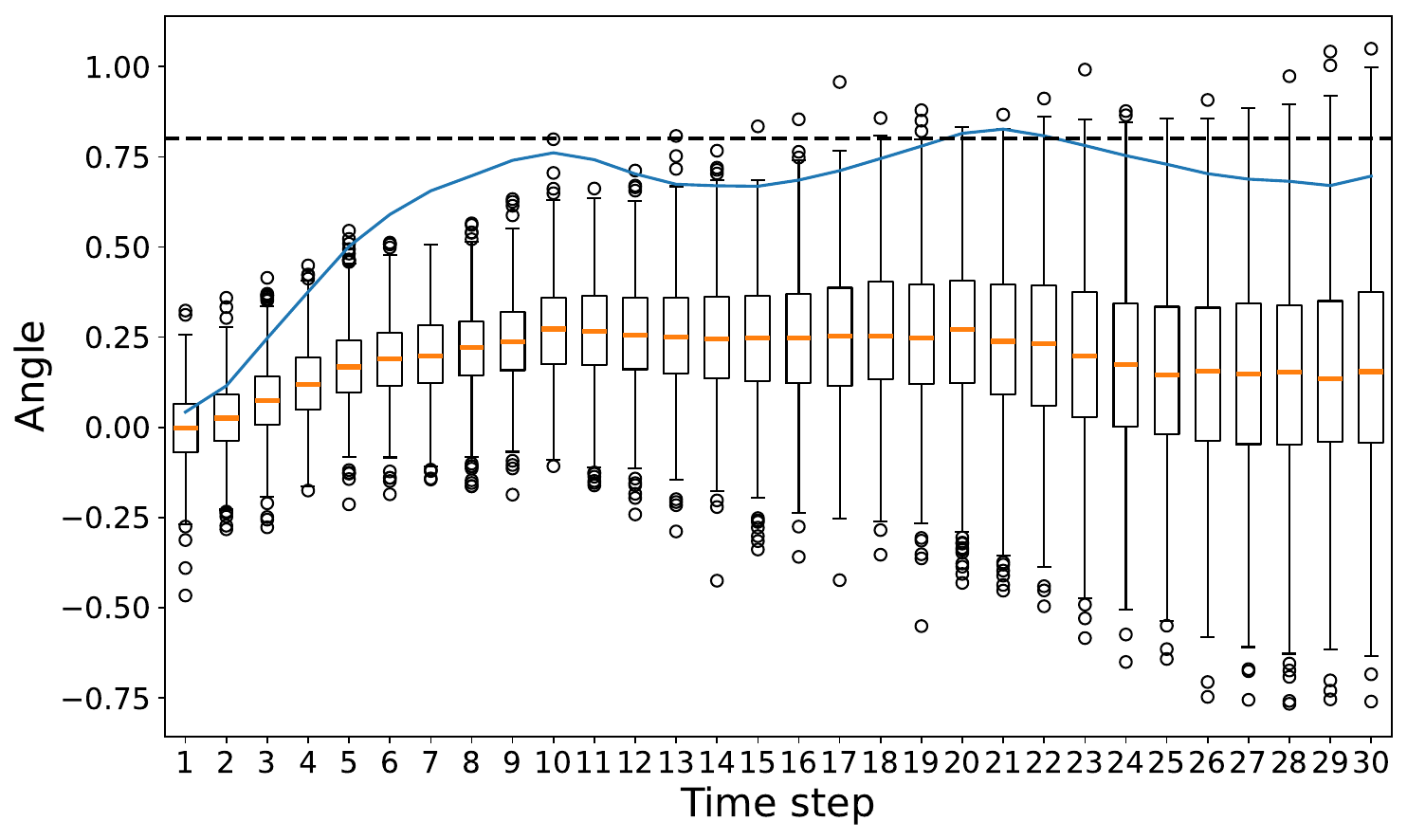}    
        \caption{DeePC}
    \end{subfigure}
    \begin{subfigure}[t]{0.49\textwidth}
        \includegraphics[width=\linewidth]{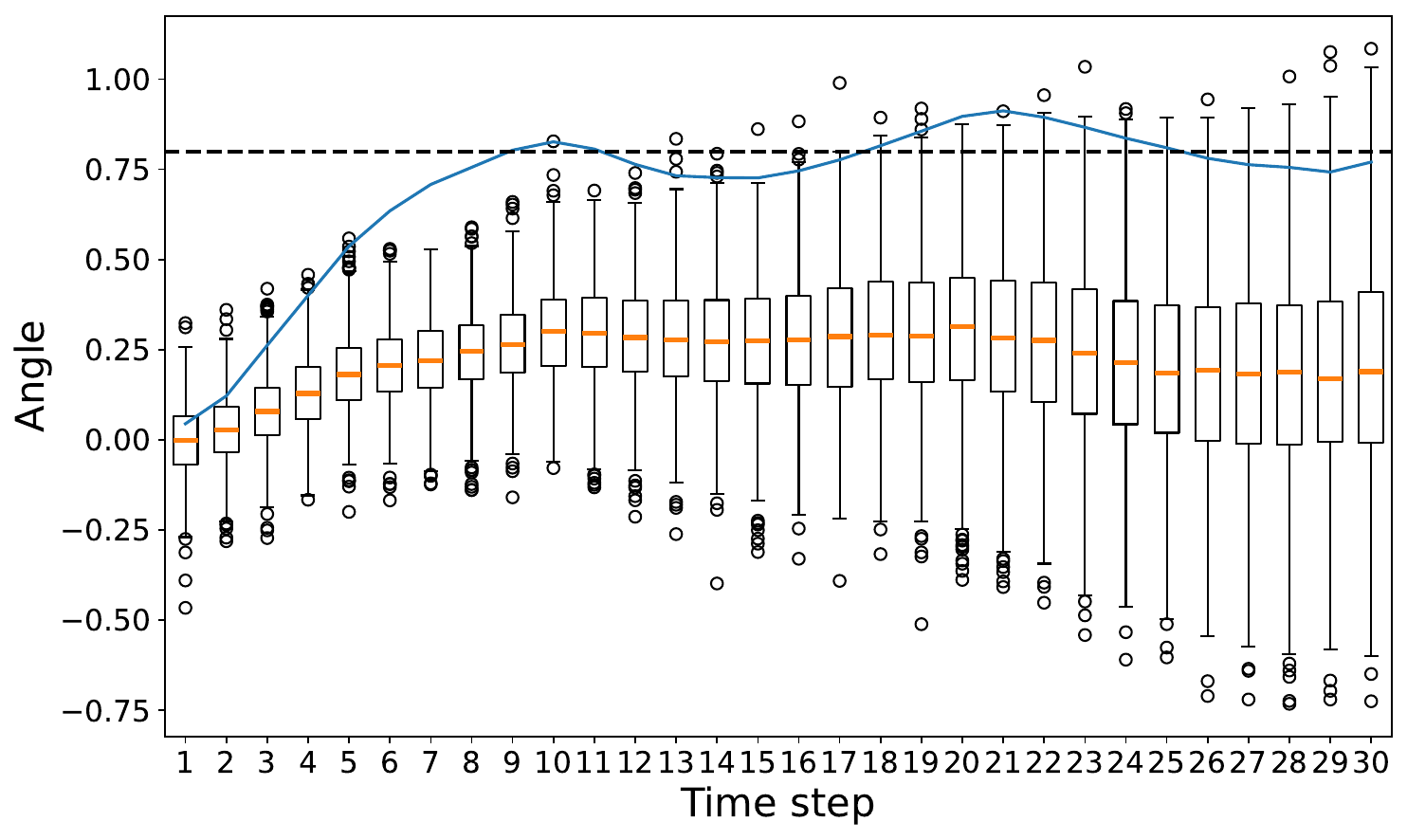}
        \caption{SPC}
    \end{subfigure}

    \begin{subfigure}[b]{0.49\textwidth}
        \includegraphics[width=\linewidth]{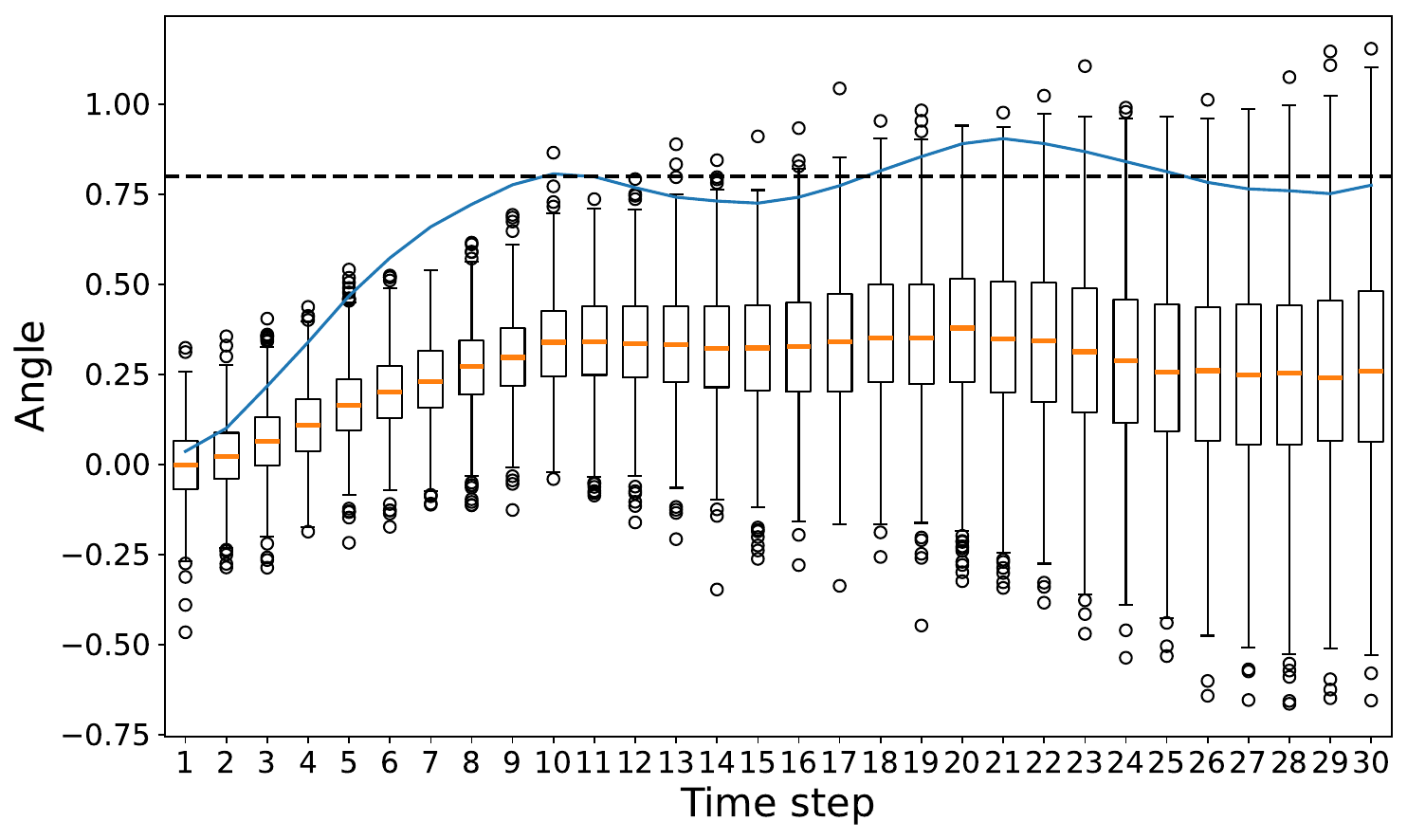} 
        \caption{MSM-SMPC}
    \end{subfigure}
    \begin{subfigure}[b]{0.49\textwidth}
        \includegraphics[width=\linewidth]{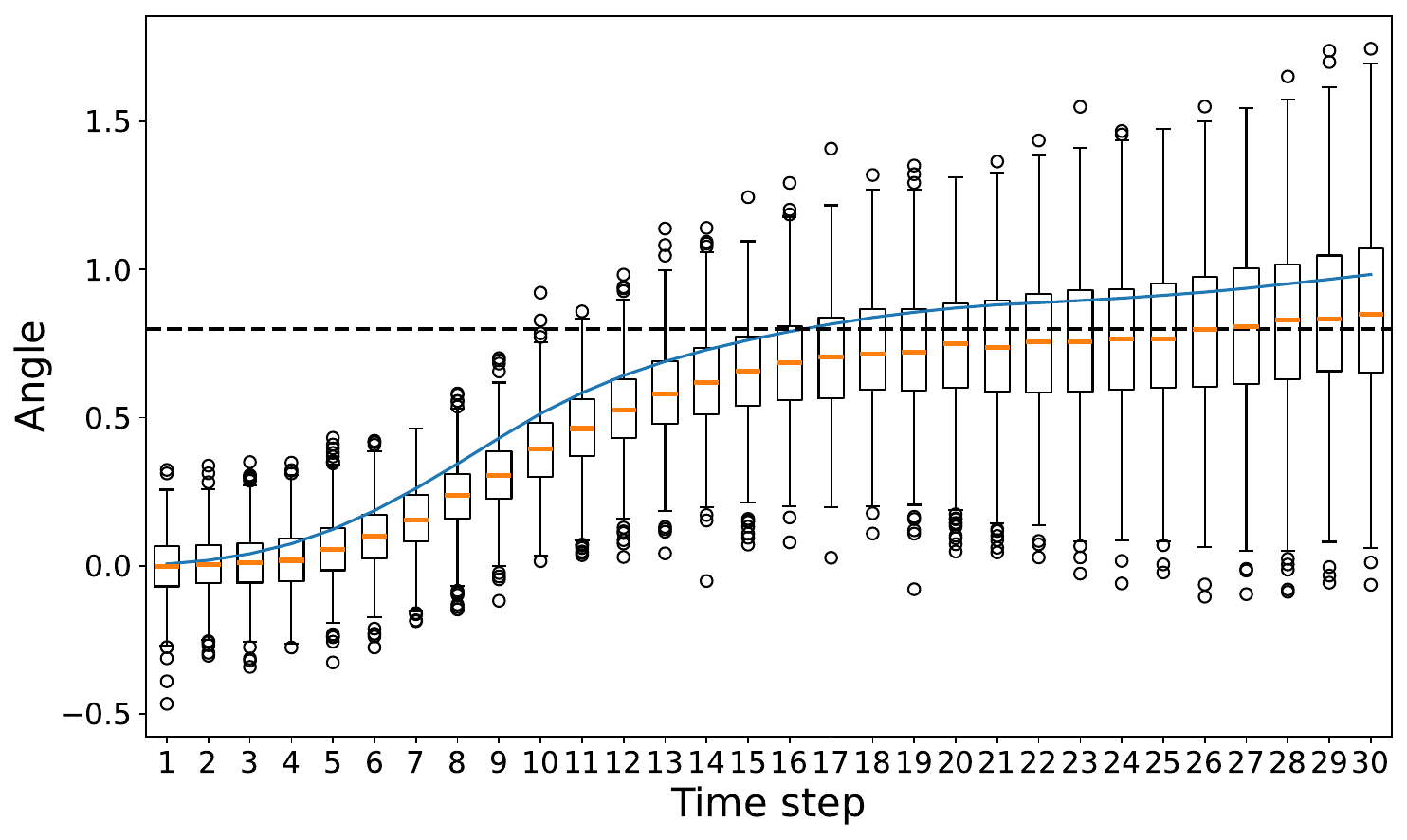}
        \caption{GB-Robust}
    \end{subfigure}
    
    \caption{Box plots of the realized output at each time step over 1000 noise realizations, for (a) DeePC, (b) SPC, (c) MSM-SMPC, and (d) GB-Robust, applied to the double spring-mass-damper system. 
    The orange lines are the realized outputs' median, the blue line is the predicted mean and the dashed line denotes the reference angle.
    By accounting for the uneven excitation in the data, GB-Robust acts cautiously, and thus, its realized outputs stay close to the predicted mean with small variance, whereas the other methods deviate substantially.}
    \label{fig:DSMD_output}
\end{figure*}
We demonstrate the proposed framework on two case studies.
First, we consider a double spring-mass-damper system with unevenly excited offline data.
We compare our robust formulation~\eqref{prb:robust_control} against \gls{SPC}, which uses only the mean prediction, \gls{DeePC}, which treats uncertainty optimistically, and the method in~\cite{fiedler2023probabilistic}, which accounts for epistemic uncertainty in a risk-agnostic manner. 
Second, we consider a building control example, where we illustrate how uncertain weather forecasts are incorporated into the predictive distribution following Section~\ref{sec:forecast}, and evaluate the chance-constrained formulations of Section~\ref{sec:CE_control} with and without disturbance feedback policies.
All methods were implemented in Python using CVXPY~\cite{diamond2016cvxpy} with MOSEK~\cite{mosek}, and the code is available online at \url{https://gitlab.ethz.ch/asasfi/gaussian_behaviors}.

\subsection{Double spring-mass-damper system}
We consider the discrete-time double spring-mass-damper system from~\cite{kerz2023data}, which consists of two rotating discs. 
The system states are the angles and angular velocities of the two discs, yielding a fourth-order system, with the angle of the first disc as the measured output.
The system has two inputs, namely torques applied to the first and the second discs, respectively.
Process noise enters the system as torques acting on both discs, and the output is corrupted by additive measurement noise. 
Both noise sources are zero-mean Gaussian with covariance $0.01\cdot I$.
The initial and prediction horizon lengths are set to $\Ti  = 4$ and $\Tf = 30$.
The control objective is to track a constant reference of $0.8$ rad from zero initial conditions, with cost weighting matrices $Q=I$ and $R=0.5\cdot I$.
The continuous-time model is discretized with sampling time $0.1$ second, and offline data are collected by simulating the system for $T=500$ time steps starting from zero initial conditions under a zero-mean Gaussian input.
Notably, the two inputs in the offline data are excited with significantly different intensities: the first input has variance $I$, while the second has variance $25\cdot I$.

We compare our robust formulation~\eqref{prb:robust}, termed GB-Robust, against three methods from the data-driven control literature: \gls{DeePC}~\eqref{eq:DeePC}, \gls{SPC}~\eqref{eq:SPC}, and the method from~\cite{fiedler2023probabilistic}, termed MSM-SMPC, which accounts for epistemic uncertainty.
MSM-SMPC minimizes the expected control cost subject to the data-dependent variance of the prediction, similarly to the confidence bound~\eqref{eq:def_M}, leading to
\begin{align}
\begin{split}
    \min_{\uf} \quad  & \|\uf - \ur\|_R^2 + \left\|\hat{M}\begin{bmatrix}
        \wi \\ \uf
    \end{bmatrix} - \yr \right\|_Q^2 \\
    & + \mathrm{tr}\left(Q\Sph\left(1+\left\|\begin{bmatrix}
        \wi \\ \uf 
    \end{bmatrix}\right\|_P^2\right)\right).
\end{split}
\end{align}
The final term in the cost penalizes the trace of the predicted output covariance weighted by $Q$.
In our robust method, we set $\epsilon = 0.1$, and the parameter $\lambda$ to be two times the lower bound, that is, $\lambda = 2\|Q \Sph\|_2$.
Furthermore, we use the efficient implementation of \gls{DeePC} from~\cite{sasfi2025gaussian}, and set $\lambda=50$.

To evaluate the methods, we start from zero initial conditions and solve each of the four optimization problems once.
We then apply the resulting input sequence to the system and simulate $1000$ independent realizations of the process and measurement noise.
Figure~\ref{fig:DSMD_output} shows box plots of the realized output at each time step, together with the predicted mean.
The control objective is to track the angle of the first disc, which is most efficiently achieved using the first actuation channel.
Accordingly, \gls{DeePC}, \gls{SPC}, and MSM-SMPC apply large torques on the first disc, as seen in Figure~\ref{fig:DSMD_input}.
However, the first actuation channel was excited far less than the second in the offline data, so predictions based on large torques on the first disc have high epistemic uncertainty.
Opposed to \gls{DeePC} and \gls{SPC}, GB-Robust accounts for this and reaches the target angle predominantly through applying torque on the second disc.
MSM-SMPC also considers this epistemic uncertainty, but to a much lesser extent.
This difference is reflected in the output trajectories in Figure~\ref{fig:DSMD_output}.
For GB-Robust, the realized outputs stay close to the predicted mean, since the controller acts cautiously.
For the other methods, the predictions are notably worse, as the realized outputs deviate substantially from the predicted mean.

We now examine the realized control costs, averaged over the $1000$ noise realizations.
\gls{SPC} attains a mean cost of $11.8$ with standard deviation $4.4$, and MSM-SMPC a mean cost of $10.3$ with standard deviation $3.9$.
The performance of \gls{DeePC} and GB-Robust depends on their tuning parameters, $\lambda_g$ and $\lambda$ respectively, whose effect on the realized cost is reported in Table~\ref{tab:costs}.
\gls{DeePC} improves as $\lambda_g$ increases, converging to \gls{SPC} for large $\lambda_g$, consistent with the optimistic interpretation of Section~\ref{sec:uncertain_control}.
Conversely, GB-Robust improves as $\lambda$ decreases, attaining its lowest cost at $\lambda = 130$, close to the lower bound $\|Q\Sph\|_2$.
Overall, GB-Robust achieves a lower realized cost than all competing methods.
\begin{table}
\setlength{\tabcolsep}{4pt}
\begin{tabular}{l|c|c|c|c|c} 
\multicolumn{6}{c}{\gls{DeePC}} \\
\hline
$\lambda_g$ & 10 & 30 & 50 & 100 & 200 \\
\hline
Cost & $15.7 \pm 5.4$ & $13.2 \pm 4.8$ & $12.7 \pm 4.6$ & $12.2 \pm 4.5$ & $12.0 \pm 4.4$ \\
\hline \\
\multicolumn{6}{c}{GB-Robust} \\
\hline
$\lambda$ & 130 & 160 & 200 & 250 & 300 \\
\hline
Cost & $6.0 \pm 1.2$ & $6.1 \pm 1.2$ & $6.2 \pm 1.2$ & $6.4 \pm 1.3$ & $6.5 \pm 1.3$ \\
\end{tabular}
\caption{Realized control cost (mean $\pm$ one standard deviation over $1000$ realizations) for \gls{DeePC} and GB-Robust, for varying tuning parameters $\lambda_g$ and $\lambda$.
The parameter-free \gls{SPC} and MSM-SMPC attain $11.8 \pm 4.4$ and $10.3 \pm 3.9$, respectively. \label{tab:costs}}
\end{table}

\begin{figure}[t!]
    \centering
    \begin{subfigure}[t]{\linewidth}
        \includegraphics[width=\linewidth]{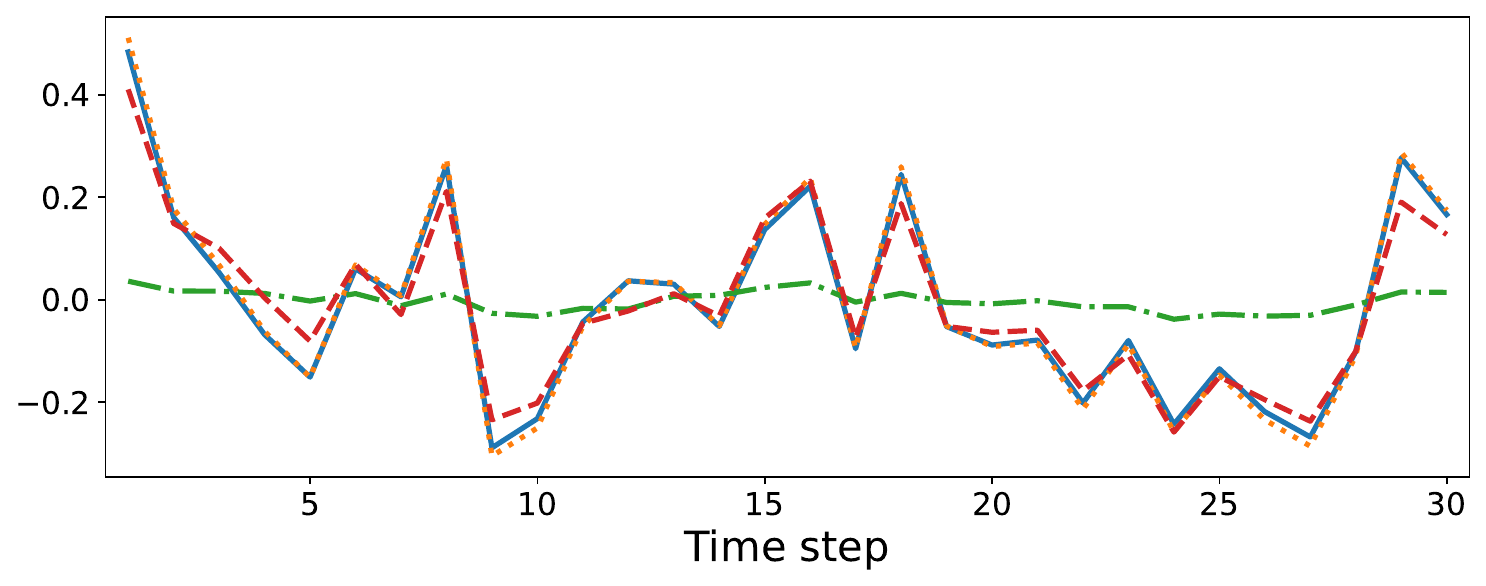}    
        \caption{Torque on first disc}
    \end{subfigure}
    \begin{subfigure}[t]{\linewidth}
        \includegraphics[width=\linewidth]{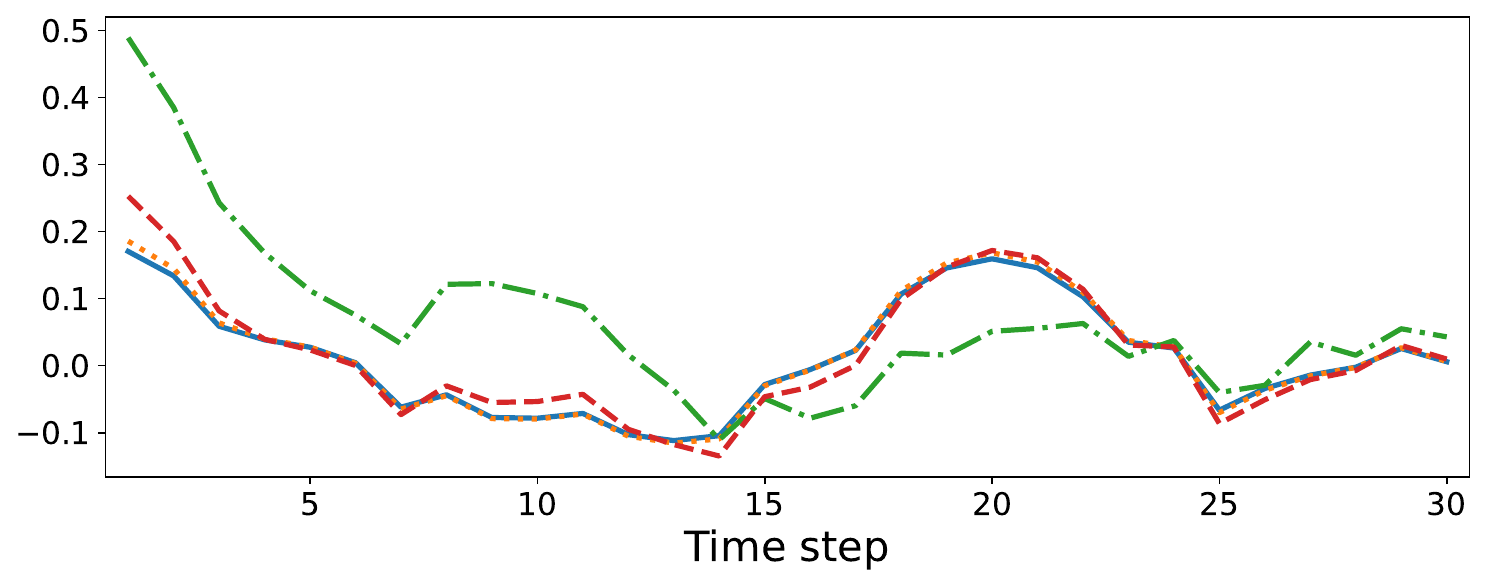}
        \caption{Torque on second disc}
    \end{subfigure}
    \caption{Optimal input sequences for the torque on the (a) first and (b) second disc, for DeePC (blue solid), SPC (orange dotted), MSM-SMPC (red dashed), and GB-Robust (green dash-dotted), applied to the double spring-mass-damper system.
    The first input channel is much less excited than the second in the offline data leading to uneven epistemic uncertainty.
    Our GB-Robust method accounts for this and reaches the target angle predominantly through applying torque on the second disc.}
    \label{fig:DSMD_input}
\end{figure}

\subsection{HVAC building control} \label{sec:HVAC}
We compare the performance of the control methods from Sections~\ref{sec:CE_feedforward} and~\ref{sec:dist_feedback} on a case study involving a \gls{HVAC} system.
We consider the simple building model adapted from~\cite{gwerder2005predictive}
\begin{align*}
    \dot{x}(t) = Ax(t) + B_u u(t) + B_v v(t),
\end{align*}
where the state $x(t) \in \mathbb{R}^7$ consists of the room temperature, and 3-3 temperature values of the thermal knots in the inner parts and the outside envelope of the room.
The system has two inputs, i.e., $u(t)\in\mathbb{R}^2$. 
The first input is the heating power, acting on the room directly, and the second is the cooling power, acting on the ceiling thermal knot.
Disturbance acts on the system in the form of internal and secondary heat gains, and varying outside air temperature.
We build a simulator of the system by discretizing the continuous-time model with Euler's method, using a sampling time of 1 minute.
The sampling time for the controller is 30 minutes.
The inputs must be positive, and the room temperature must be between $21$ and $26$ degrees with probability $99\%$.

We assume that an uncertain forecast is available for the outside air temperature, and incorporate this information into the prediction, as described in Section~\ref{sec:forecast}.
Thus, the estimated predictive distribution takes into account the uncertainty of the prediction.
We model the $k$-th entry of the forecast error $e_v$ as the cumulative sum $\sum_{j=1}^{k} e_j$ of i.i.d. increments $e_j \sim \mathcal{N}(0,0.3^2)$.
Equivalently, the forecast of the temperature $k$ steps ahead is corrupted by the accumulated noise of all previous steps.
This model resembles real forecasts, as it captures growing uncertainty over the horizon, while remaining smooth.

We compare the two formulations from Section~\ref{sec:CE_control}, namely the chance-constrained problem~\eqref{eq:stoch_control} with a feedforward input sequence, and its extension~\eqref{prb:feedback_policy}, which optimizes over disturbance affine feedback policies.
The two formulations differ only in the input parameterization, that is,~\eqref{eq:stoch_control} corresponds to~\eqref{prb:feedback_policy} with $K=0$ fixed.
We first compare the open-loop predictions.
Figures~\ref{fig:open_loop-nofb} and~\ref{fig:open_loop-fb} show the predicted mean (solid blue line) and the confidence bounds computed from $\Sph$ (shaded blue region) for the feedforward and the feedback policies, respectively.
To validate the predictions, we apply the same optimal input policy to the system and simulate $100$ noise realizations, depicted in orange.
As a cold day is simulated, the optimal cooling power is zero for both methods.
With the feedforward input, the predicted uncertainty of the room temperature accumulates over the horizon, and the controller must keep the predicted mean far from the lower comfort constraint to satisfy the chance constraints, see Figure~\ref{fig:open_loop-nofb}. 
In contrast, the feedback policy reacts to realized prediction errors and thereby reduces the output variance, allowing the controller to plan trajectories closer to the constraint without sacrificing safety. 
As a trade-off, the inputs become stochastic, since the applied heating power depends on the realized errors, which is visible in the spread of the input realizations in Figure~\ref{fig:open_loop-fb}.
Lastly, note that the methods of Section~\ref{sec:CE_control} only account for aleatoric uncertainty.
As a result, the predicted and true output distributions do not match exactly, which is visible in the realizations slightly exceeding the confidence bounds.

Next, we compare the closed-loop performance of the two formulations.
Both controllers are applied in a receding-horizon fashion, as in MPC, that is, at each sampling instant the optimization problem is solved with the most recent measurements, and the first input of the optimal policy is applied to the system. 
The simulation spans two days, during which the outside temperature follows a day-night cycle with rising overall temperatures.
The resulting closed-loop trajectories are shown in Figure~\ref{fig:closed_loop}.
During warm periods, the room temperature approaches the upper comfort constraint, and at night it drops toward the lower one.
Since the disturbance feedback policy can safely plan trajectories closer to the constraints, it requires less preemptive heating and cooling than the feedforward formulation.
As a result, optimizing the feedback gain reduces the closed-loop control cost, which mostly reflects the actual heating and cooling effort, by $77\%$.

\begin{figure}
    \centering
    \begin{subfigure}[]{\linewidth}
        \includegraphics[width=\linewidth]{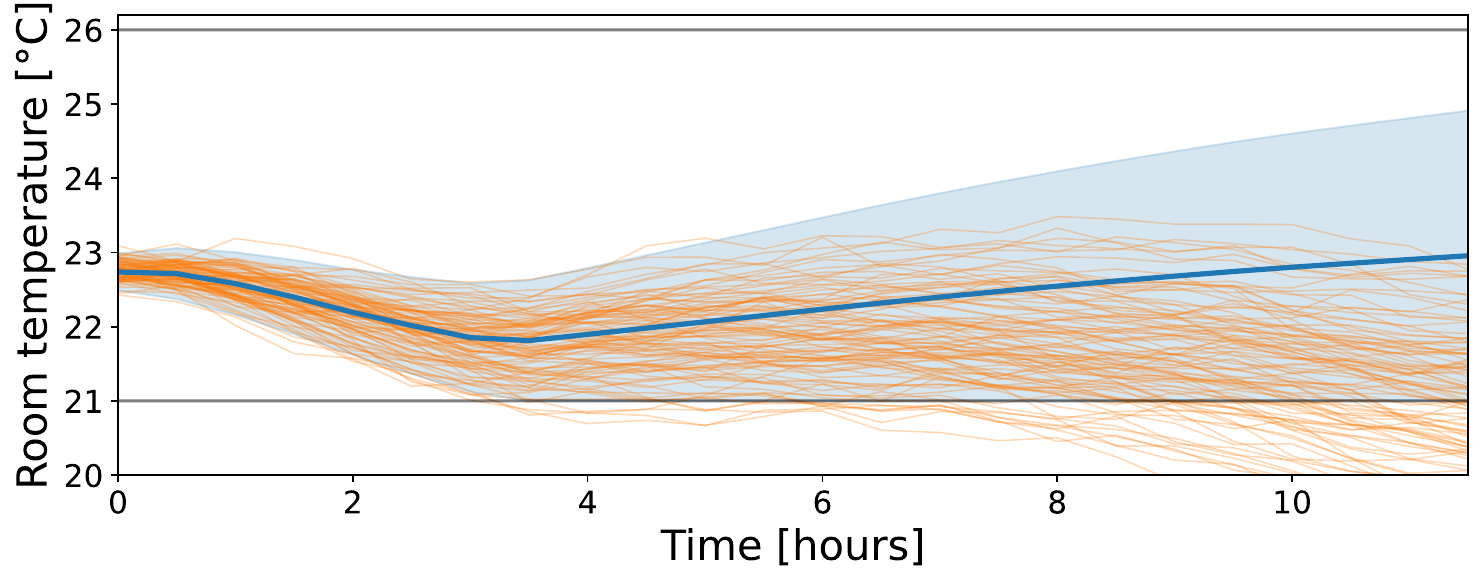}
    \end{subfigure}
    \begin{subfigure}[]{\linewidth}
        \includegraphics[width=\linewidth]{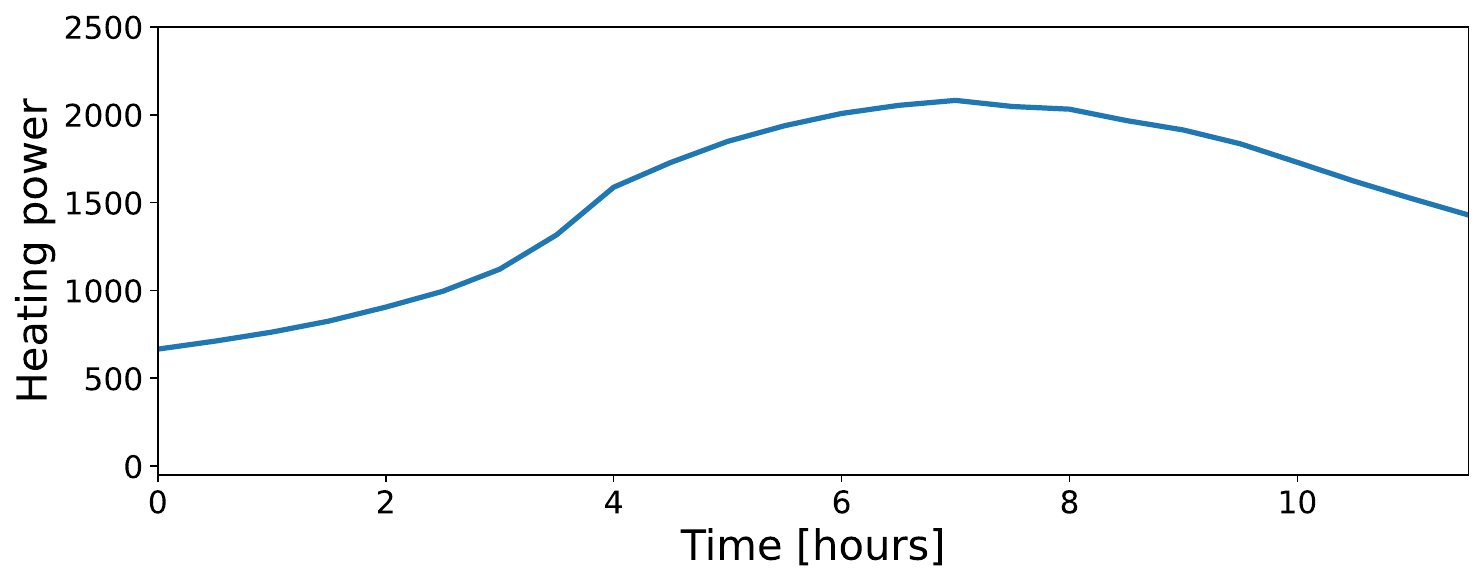}
    \end{subfigure}
    \caption{Open-loop prediction with the feedforward input policy~\eqref{eq:stoch_control}. Top: predicted mean room temperature (blue solid line) and confidence bounds ($95\%$) computed from $\Sph$ (blue shaded region), together with $100$ realizations of the true room temperature (orange thin lines).
    Bottom: heating power applied to the system.}
    \label{fig:open_loop-nofb}
\end{figure}

\begin{figure}
    \centering
    \begin{subfigure}[]{\linewidth}
        \includegraphics[width=\linewidth]{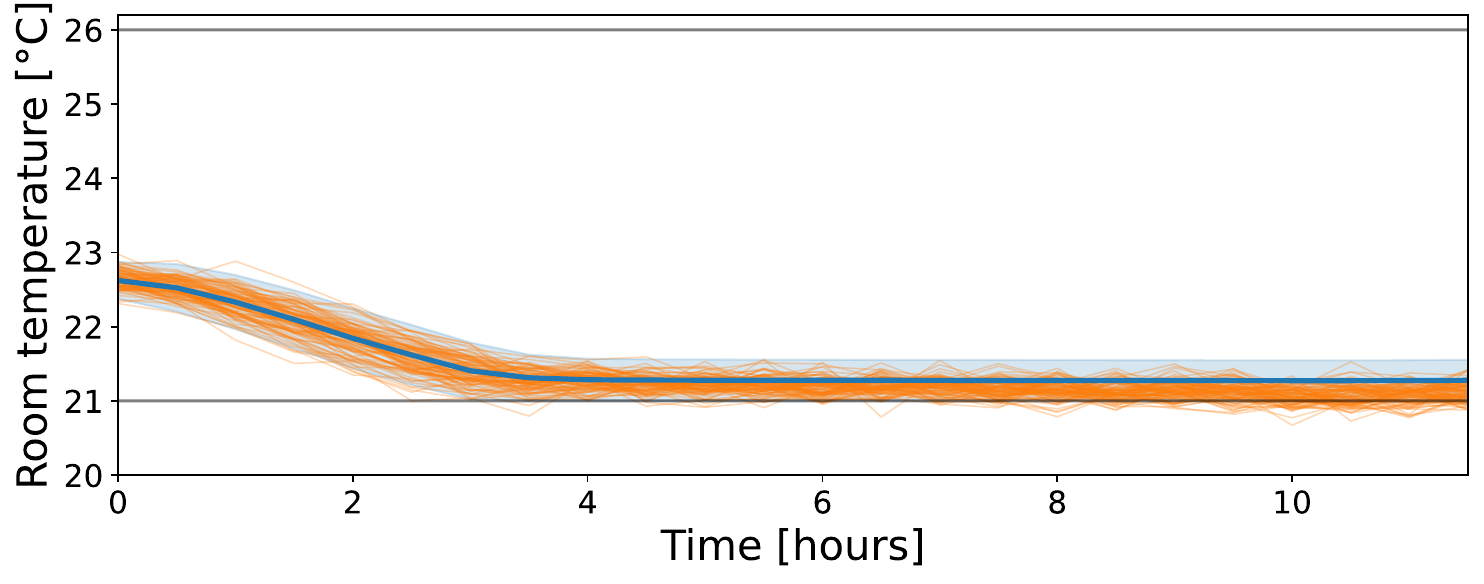}
    \end{subfigure}
    \begin{subfigure}[]{\linewidth}
        \includegraphics[width=\linewidth]{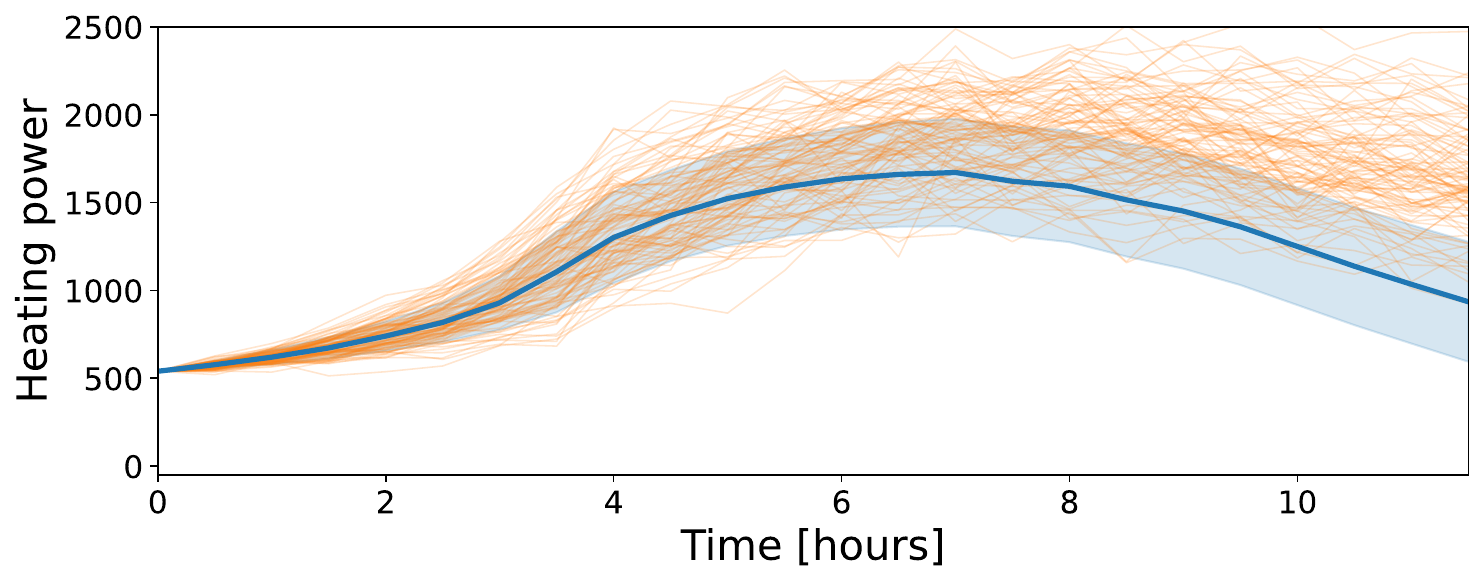}
    \end{subfigure}
    \caption{Open-loop prediction with the optimized disturbance affine feedback policy~\eqref{prb:feedback_policy}. Top: predicted mean room temperature (blue solid line) and confidence bounds ($95\%$) computed from $\Sph$ (blue shaded region), together with $100$ realizations of the true room temperature (orange thin lines).
    Bottom: predicted mean (blue solid line) and realizations (orange thin lines) of the heating power, which is stochastic due to the disturbance feedback.}
    \label{fig:open_loop-fb}
\end{figure}

\begin{figure}
    \centering
    \begin{subfigure}{\linewidth}
        \includegraphics[width=\linewidth]{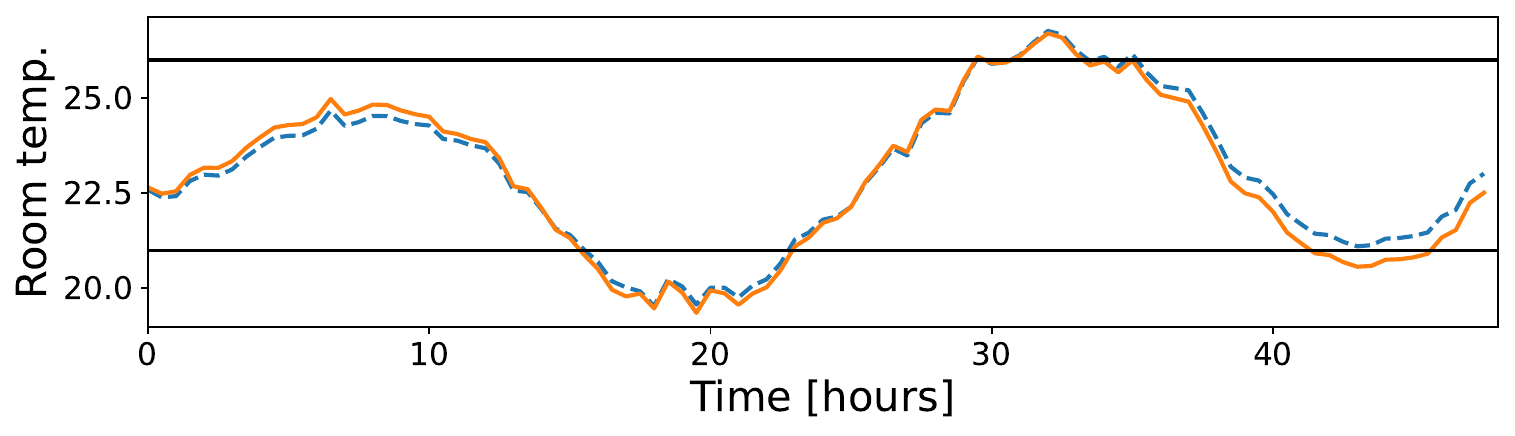}
    \end{subfigure}
    \begin{subfigure}{\linewidth}
        \includegraphics[width=\linewidth]{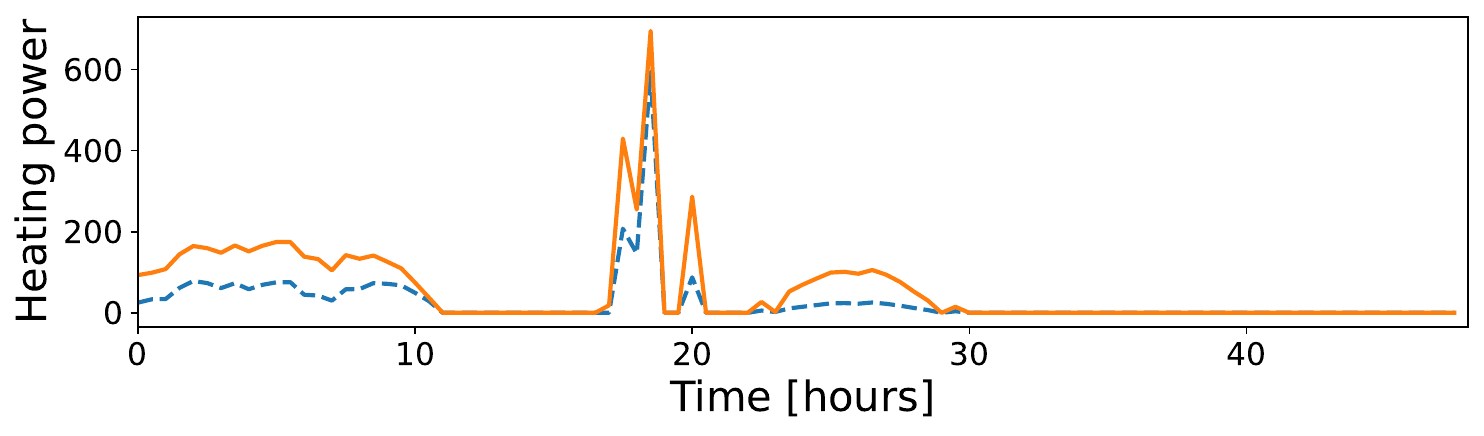}
    \end{subfigure}
    \begin{subfigure}{\linewidth}
        \includegraphics[width=\linewidth]{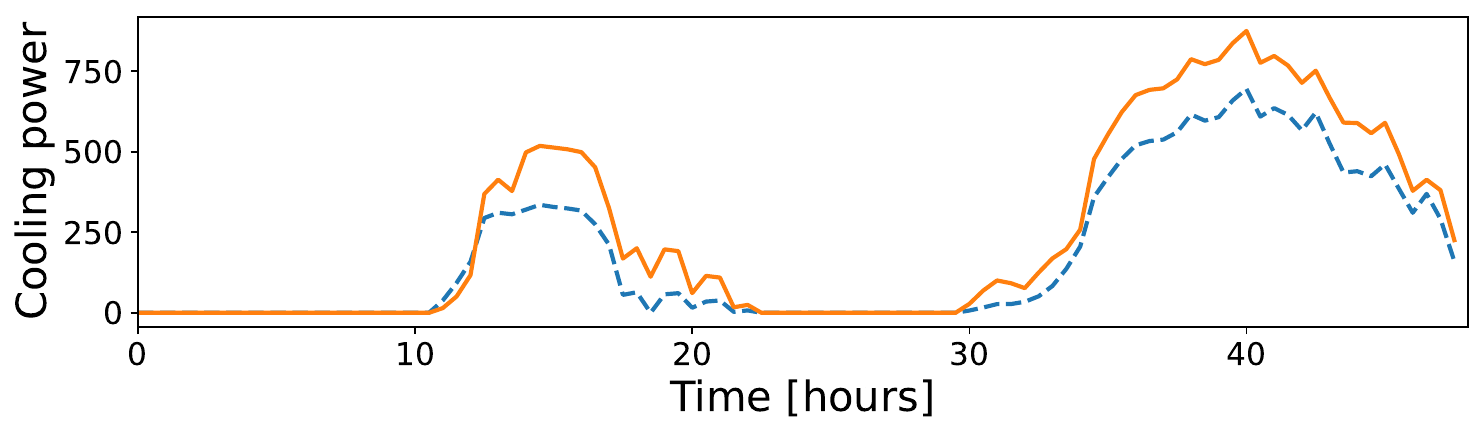}
    \end{subfigure}
    \caption{Closed-loop operation over two days with the feedforward policy~\eqref{eq:stoch_control} (orange solid) and the disturbance feedback policy~\eqref{prb:feedback_policy} (blue dashed). Top: room temperature and comfort constraints (black lines). Middle and bottom: applied heating and cooling power, respectively.}
    \label{fig:closed_loop}
\end{figure}

\section{Conclusion} \label{sec:conclusion}
We introduced Gaussian behaviors, a stochastic extension of behavioral systems theory that augments a deterministic \gls{LTI} behavior with a Gaussian noise component.
From this framework we derived a data-driven method for prediction which uses the sample covariance of trajectory data and delivers finite-sample bounds on the prediction uncertainty.
Building on this method, we proposed stochastic data-driven predictive control schemes, including a chance-constrained formulation over disturbance affine feedback policies and a robust formulation that minimizes the worst-case cost over a data-driven confidence region, both of which admit tractable convex reformulations. 
Future work includes relaxing the independence assumption on the offline data trajectories to accommodate more data-efficient formulations, and developing online updates of the Gaussian behavior to enable adaptive data-driven control.

\section*{References}
\bibliographystyle{ieeetr}        
\bibliography{Literature}

\end{document}